\documentclass[10pt,reqno]{article}
\usepackage{graphicx,amsmath,amssymb,setspace}
\usepackage{natbib}

\newtheorem {prop}{Proposition}[section]
\newtheorem{assumption}{Assumption}
\newtheorem{ppt}{Property}[section]
\newtheorem{thm}{Theorem}[section]
\newtheorem{lemm}{Lemma}[section]
\newtheorem{remark}{Remark}[section]
\newtheorem{cor}{Corollary}[section]

\newenvironment{proof}{\textsc{Proof:}}{\mbox{ } \hfill $\Box$ \vspace{2mm}}
\numberwithin{equation}{section}

\renewcommand{\d}{\mathrm{d}}

\newcommand{\R}{\mathbb R}

\newcommand{\E}{\mathbb E}

\newcommand{\F}{\mathbb F}

\newcommand{\G}{\mathcal{G}}

\newcommand{\lo}{\lambda^o}

\newcommand{\D}{\mathfrak{D}}

\newcommand{\li}{\lambda^I} %Special care is needed to use this abbreviation
\newcommand{\lp}{\lambda^P} %Special care is needed to use this abbreviation
\newcommand{\uli}{\underline{\lambda}^I}
\newcommand{\ulp}{\underline{\lambda}^P}
\newcommand{\Pn}{P^{(n)}}
\newcommand{\psin}{\psi^{(n)}}
\renewcommand{\L}{\mathcal{L}}
\newcommand{\qr}{q^r}
\newcommand{\qli}{q^{\li}}

\newcommand{\pili}{\pi^{\li}}
\newcommand{\lia}{\left(\li-\uli \right)}
\newcommand{\lpa}{\left(\lp-\ulp \right)}

\renewcommand{\wr}{W^r}
\newcommand{\wi}{W^{I}}

\newcommand{\fn}{f^{(n)}}
\newcommand{\phin}{\phi^{(n)}}
\newcommand{\dd}{\displaystyle}
\renewcommand{\l}{\lambda}
\newcommand{\0}{{\bf0}}
\newcommand{\1}{{\bf1}}
\renewcommand{\b}{\beta}
\newcommand{\gamman}{\gamma^{(n)}}
\newcommand{\Gamman}{\Gamma^{(n)}}

\begin{document}

\title{Hedging Pure Endowments with Mortality Derivatives}

\author{Ting Wang % <-this % stops a space
\thanks{Department of Mathematics, University of Michigan, Ann Arbor, MI 48109,  email: wting@umich.edu.}
%\thanks{E. Bayraktar is supported in part by the National Science Foundation under grant DMS-0604491.}
\and Virginia R. Young \thanks{ Department of Mathematics,
University of Michigan, Ann Arbor, Michigan, 48109, email:vryoung@umich.edu.  V.\ R.\ Young thanks the Nesbitt Professorship of Actuarial Mathematics for financial support.} }

\date{30 September 2010}

\maketitle

\begin{abstract}
In recent years, a market for mortality derivatives began developing as a way to handle systematic mortality risk, which is inherent in life insurance and annuity contracts.  Systematic mortality risk is due to the uncertain development of future mortality intensities, or {\it hazard rates}. In this paper, we develop a theory for pricing pure endowments when hedging with a mortality forward is allowed.  The hazard rate associated with the pure endowment and the reference hazard rate for the mortality forward are correlated and are modeled by diffusion processes.  We price the pure endowment by assuming that the issuing company hedges its contract with the mortality forward and requires compensation for the unhedgeable part of the mortality risk in the form of a pre-specified instantaneous Sharpe ratio. The major result of this paper is that the value per contract solves a linear partial differential equation as the number of contracts approaches infinity.  One can represent the limiting price as an expectation under an equivalent martingale measure.  Another important result is that hedging with the mortality forward may raise or lower the price of this pure endowment comparing to its price without hedging, as determined in \cite{BayraktarMilevskyPromislow2008Valuation}.  The market price of the reference mortality risk and the correlation between the two portfolios jointly determine the cost of hedging.  We demonstrate our results using numerical examples.

\emph{Keywords.} Life annuities, longevity risk, $q$-forward, mortality-linked derivatives, instantaneous Sharpe ratio, incomplete market

\end{abstract}

%-----------------------------------------------------------------------------------------------------------------------------------------------------------------------------------------------------------------------------------------------------

\section{Introduction and Motivation}

A basic assumption in many actuarial texts is that mortality risk can be eliminated based on the law of large number.  It is believed that the standard deviation per insurance policy vanishes as the number of policies sold becomes large enough.  However, this assumption is valid only when the mortality intensity is deterministic, and a number of recent researchers argue that mortality intensity, or {\it hazard rate}, is stochastic; see, for example, \cite{DowdBlakeCairnsDawson2006} and the references therein. The uncertainty of hazard rates is significant enough that stochastic mortality risk has to be considered in the valuation of life insurance and annuity contracts and in pension fund management.  A concrete example of stochastic mortality risk is longevity risk, namely, the risk that future lifetimes will be greater than expected.  Longevity risk has attracted much attention in recent years, and many capital market instruments have been proposed  to deal with this risk for annuity providers and pension funds; see \cite{DowdBlakeCairnsDawson2006}, \cite{BlakeBurrows2001}, and \cite{blake2006living} for more details.  

However, few researchers have focused on the effectiveness of hedging mortality risk with the proposed mortality-linked derivatives; one notable exception is the work of \cite{LinCox2005Securitization}.  In this paper, we investigate the application of mortality-linked derivatives for hedging mortality risk and offer suggestions for further mortality-linked innovation based on our analysis. To this end, we select a stochastic model to describe the mortality dynamics. Several stochastic mortality models have been proposed in the recent literature.  \cite{MilevskyPromislow2001}, \cite{Biffis2005Affine}, \cite{Schrager2006Affine}, and \cite{Dahl2004} use continuous-time diffusion processes to model the hazard rate, as we do in this paper.  Alternatively, \cite{MiltersenPersson2005} and \cite{CairnsBlakeDowd2006} model the forward mortality.  \cite{MilidonisLinCox2010} incorporate mortality state changes into the mortality dynamics with a discrete-time Markov regime-switching model. Also, see \cite{CairnsBlakeDowd2006Pricing} for a detailed overview of various modeling frameworks.   In this paper, we use the model proposed by  \cite{BayraktarMilevskyPromislow2008Valuation} to describe the dynamics of both hazard rates: $\lp_t$, the one inherent in the insurance contract to be hedged, and $\li_t$, the one referenced by the mortality-linked derivative. 

Another issue is the choice of pricing paradigm.  Different methods for pricing mortality risk have been proposed in recent literatures, and \cite{bauer2009pricing} extensively discusses them.  Among these methods,  \cite{BayraktarMilevskyPromislow2008Valuation}  developed a dynamic pricing theory, which can be considered as a continuous version of the actuarial standard deviation premium principle.  In our paper, we extend their pricing mechanism to a market that includes mortality-linked derivatives.  We price a pure endowment assuming that the issuing company hedges its contract with a mortality forward in order to minimize the variance of the value of the hedging portfolio and then requires compensation for the  unhedgeable part of the mortality risk in the form of a pre-specified instantaneous Sharpe ratio.  

The main purpose of this paper is to investigate the hedging of life insurance and annuity contracts with mortality-linked derivatives.  To this end,  we develop a partial differential equation (PDE) whose solution is the value of the hedged insurance contract.  We compare the values of the hedged contract under different market prices of mortality risk.  We also analyze how the correlation between $\lp_t$ and $\li_t$ affects the values of the hedged contract. The main contribution of our paper is to show that hedging can reduce the price of the insurance contract only under certain conditions on the correlation of the hazard rates and on the market price of mortality risk.  As part of the procedure, we also show that the desired features of the pricing mechanism by \cite{MilevskyPromislowYoung2005Financial} and \cite{BayraktarMilevskyPromislow2008Valuation} still hold in our extension. 

The remainder of this paper is organized as follows:  In Section \ref{sec:market}, we present our financial market, describe the pricing mechanism of the pure endowment in a market with mortality-linked derivatives, and derive a non-linear PDE whose solution is the value of the hedged pure endowment.  In Section \ref{sec:prop_Pn}, we analyze the value $P^{(n)}$ of $n$ pure endowments on conditionally independent and identically distributed lives, with the emphasis on how  the correlation of the hazard rates and the market price of mortality risk affect the price of the hedged pure endowments.  We then present the PDE that gives the limiting value of $\frac{1}{n}P^{(n)}$  as $n$ goes to infinity in Section \ref{sec:limit}.  We show that this limiting value solves a {\it linear} PDE and represent this value as an expectation with respect to an equivalent martingale measure.  In Section \ref{sec:numerical}, we demonstrate our results with numerical examples, discuss whether and when the hedging with mortality-risk derivatives reduces the price of pure endowments, and provide suggestions on the application of mortality-linked derivatives for insurance companies.  We describe a numerical scheme to compute the value of a pure endowment in Section \ref{sec:appendix}.  Section \ref{sec:conclusion} concludes the paper.

%-----------------------------------------------------------------------------------------------------------------------------------------------------------------------------------------------------------------------------------------------------
\section{\label{sec:market} Incomplete Market of Financial and Mortality Derivatives}

In this section, we describe the pure endowment contract and the financial market in which the issuer of the contract invests  to hedge the risk.  In the financial market, there are three products:  a money market fund, a bond, and a mortality derivative.  We obtain the optimal strategy to hedge the risk of the contract with bonds and mortality derivatives in order to minimize the variance of the value of the investment portfolio.  We, then, price the pure endowment using the instantaneous Sharpe ratio.

%-----------------------------------------------------------------------------------------------------------------------------------------------------------------------------------------------------------------------------------------------------

\subsection{Mortality Model and Financial Market}
First, we set up the model for the dynamics of hazard rates--either the hazard rate for the pure endowment or the one for the mortality derivative.  We assume that a hazard rate $\lambda_t$ follows a diffusion process with some positive lower bound $\underline{\lambda}$.  Thus, we require that as $\lambda_t$ goes to $\underline{\lambda}$, the drift of $\lambda_t$ is positive and the volatility of $\lambda_t$ approaches $0$. Biologically, the lower bound $\underline{\lambda}$ represents the remaining hazard rate after all accidental or preventable causes of death have been removed.  Mathematically, the need for such a lower bound appears later in this paper.

Specifically, we use the following diffusion model for a hazard rate:
\begin{equation}
\label{eqn_l}
\d \l_t=a(\l_t,t)\left(\l_t-\underline{\l} \right)\d t+b(t) \left(\l_t-\underline{\l} \right)\d W_t,
\end{equation}
in which $W$ is a standard Brownian motion on a filtered probability space $(\Omega, {\mathbb F}, ({\mathbb F}_t)_{t \ge 0}, {\mathbb P})$.  We require that the volatility $b(t)$ is a continuous function of $t$ and is bounded from below by a positive constant $\kappa$ in $[0,T]$.  We also assume that $a(\l_t, t)$ is H\"older continuous with respect to $\l$ and $t$, and that $a(\l_t, t)>0$ when $0 <\l_t-\underline{\l}<\epsilon$ for some $\epsilon>0$.

In this paper, we consider two different but correlated hazard rates.  One is  the hazard rate of the insured population; namely, the hazard rate of the people who purchase the pure endowments.  For simplicity, when we consider a portfolio of $n$ pure endowment contracts in this paper, we assume that all individuals are of the same age and are subject to the same hazard rate.  We denote as $\lp_t$ the hazard rate of insured population, and the dynamics of $\lp_t$ is given by
\begin{equation}
\label{lp}
\d \l^P_t=a^P(\l^P_t,t)\left(\l^P_t-\ulp \right)\d t+b^P(t)\left(\l^P_t-\ulp \right)\d W^P_t. 
\end{equation}
We also consider a second hazard rate on which  the  mortality derivatives are based, namely, the hazard rate of an indexed population.  We denote this hazard rate as $\li_t$, whose dynamics is given by
\begin{equation}
\label{li}
\d \l^I_t=a^I(\l^I_t,t)\left(\l^I_t-\uli \right)\d t+b^I(t)\left(\l^I_t-\uli \right)\d W^I_t.
\end{equation}
The uncertainties of the two hazard rates are correlated such that $\d W^I_t \, \d W^P_t = \rho \, \d t$ with $\rho\in[-1,1]$. 

Suppose, at time $t=0$, an insurer issues a {\it pure endowment} to an individual that pays \$1 at time $T$ if the individual is alive at that time. To price this contract, we will create a portfolio composed of the obligation to pay this pure endowment and investment in the financial market. 

In the financial market, the dynamics of the short rate $r_t$ is given by 
\begin{equation} \label{eq:shortrate}
\d r_t=\mu(r_t,t) \, \d t +\sigma(r_t,t) \, \d W_t^r
\end{equation}
in which $\mu$ and $\sigma \ge 0$ are deterministic functions of the short rate and time, and $W^r$ is a standard Brownian motion adapted to $(\Omega, {\mathbb F}, ({\mathbb F}_t)_{t \ge 0}, {\mathbb P})$. We assume that $W^r$ is independent of $W^P$ and $W^I$, and that $\mu$ and $\sigma$ are such that $r_t > 0$ almost surely for all $t \ge 0$ and such that \eqref{eq:shortrate} has a unique solution. 

Both the $T$-bond and the mortality derivative are priced based on the principle of no-arbitrage. Thus, for the short rate $r$, there exists a market price of risk $q^r$ that is adapted to the filtration generated by $W^r$; and for the hazard rate $\li_t$, there exists a market price of risk $q^{\li}$ that is adapted to the filtration generated by $W^I$. We, therefore, write either $q^r_t $ or $q^r(r_t, t)$ for the market price of the short rate risk at time $t$; similarly, we write either $q^{\li}_t$ or $q^{\li}(\li_t, t)$ for the market price of the hazard rate risk at time $t$.

Define an equivalent martingale measure ${\mathbb Q}$ whose Radon-Nikodym derivative with respect to ${\mathbb P}$  is given by 
 \begin{equation}
\dfrac{\d {\mathbb Q}}{\d {\mathbb P}} = \exp\left\{ -\int_0^T \left[ \qr(r_s, s) \, \d \wr_s + \qli(\lambda^I_s, s) \, \d \wi_s \right] - \frac{1}{2} \int_0^T \left[  \left( \qr(r_s, s) \right)^2+\left(\qli(\lambda^I_s, s)\right)^2\right] \d s  \right \}.
\end{equation}
In the $\mathbb Q$-space, the dynamics of the hazard rates and the short rate are given by
\begin{equation}
\begin{cases}
\d \l^P_t=a^{P,Q}(\l^P_t, \li_t, t) \left(\l^P_t-\ulp \right) \d t + b^P(t) \left(\l^P_t-\ulp \right) \d W^{P,Q}_t,\\
\d \l^I_t=a^{I,Q}(\l^I_t,t) \left(\l^I_t-\uli \right) \d t + b^I(t) \left(\l^I_t-\uli \right) \d W^{I,Q}_t,\\
\d r_t=\mu^Q(r_t,t) \, \d t + \sigma(r_t,t) \, \d W_t^{r,Q},\\
\end{cases}
\end{equation}
in which
\begin{equation}
\begin{cases}
W_t^{P,Q}=W_t^{P} + \rho \int_0^t q^{\li}(\l^I_s, s) \, \d s,\\
W_t^{I,Q}=W_t^{I}+\int_0^t  q^{\li}(\l^I_s, s) \, \d s,\\  
W_t^{r,Q}=W_t^{r}+\int_0^t  q^{r}(r_s, s) \, \d s,\\
 \end{cases}
\end{equation}
and 
\begin{equation}
\begin{cases}\label{muq}
a^{P,Q}(\l^I_t, \l^P_t, t) = a^P(\l^P_t,t) -  \rho \, \qli(\l^I_t,t) \, b^P(t),\\
a^{I,Q}(\l^I_t, t) = a^I(\l^I_t, t) - \qli(\l^I_t, t) \, b^I(t), \\
\mu^Q(r_t, t) = \mu(r_t, t) - \qr(r_t, t) \, \sigma(r_t, t).
\end{cases}
\end{equation}

The time-$t$ price of the $T$-bond is given by
\begin{equation}
F(r, t ;T) = \E^Q \left[e^{-\int_t^T r_s \, ds} \bigg| \, r_t = r \right],
\end{equation}
and the bond price $F$ solves the following partial differential equation (PDE), \cite{Bjork2004Arbitrage}:
\begin{equation}\label{FK_F}
\begin{cases}
F_t +\mu^Q(r,t) \, F_r +\frac{1}{2}\sigma^2(r,t) \, F_{rr} - rF = 0, \\
F(r, T; T) = 1.
\end{cases}
\end{equation}
From this PDE, we obtain the dynamics of $F$ for $t \le s \le T$:
\begin{equation}\label{eqn_F}
\begin{cases}
\d F(r_s, s)=\left[r_sF(r_s, s) + q^r(r_s, s) \, \sigma(r_s,s) \, F_r(r_s, s)\right] \d s + \sigma(r_s, s) \, F_r(r_s, s) \, \d W^r_s\\
F(r_t, t) = F(r, t).
\end{cases}
\end{equation}

Without loss of generality, we specify the mortality derivative as a $q$-forward.  Define the cumulative hazard rate process by $\Lambda^I_t = \int_0^t \li_s \, ds$ for $0 \le t \le T$.  Then, the time-$t$ value of a $q$-forward with delivery time $T$ is given by 
\begin{equation} \label{qforward}
S\left(r, \li, \Lambda^I, t; T \right) = \E^Q\left[e^{-\int_t^T r_s \, \d s}\left(e^{-\int_0^T \li_s \, \d s}-K\right)\bigg| \, r_t = r, \li_t = \li, \Lambda^I_t = \Lambda^I \right]. 
\end{equation}
in which $K = \E^Q\left[ e^{-\int_0^T \li_t \, \d t}\bigg| \, \F_0 \right]$ is the delivery price. As for the $T$-bond, we have the following PDE for the mortality derivative:
\begin{equation}\label{eqn_S}
\begin{cases}
S_t +\mu^Q \, S_r +\frac{1}{2}\sigma^2 \, S_{rr} + a^{I, Q} \cdot \lia \, S_{\li} + \dfrac{1}{2}\left(b^I \right)^2\left(\li-\uli\right)^2 S_{\li \li} + \li \, S_{\Lambda^I} - rS = 0, \\
S(r, \li, \Lambda^I, T; T) = e^{-\Lambda^I} - K.
\end{cases}
\end{equation}
From the PDE above, we obtain the dynamics of $S$ for $t \le s \le T$:
\begin{equation}\label{FK_S}
\begin{cases}
\d S_s = \left[r_sS + q^r_s \, \sigma S_r+q^{\li}_s \, b^I \cdot \left(\li_s - \uli \right)S_{\li}\right] \d s +\sigma \, S_r \, \d W^r_s+b^I \cdot \left(\li_s - \uli \right)S_{\li} \, \d W^{I}_s, \\
S(r_t,\li_t, \Lambda^I_t, t) = S(r, \li, \Lambda^I, t).
\end{cases}
\end{equation}
Since we fixed the maturity $T$ for both the bond and the $q$-forward in this paper, we drop the notation $T$ when appropriate.  

%-----------------------------------------------------------------------------------------------------------------------------------------------------------------------------------------------------------------------------------------------------
\subsection{Pricing the Pure Endowment via the Instantaneous Sharpe Ratio }
%-----------------------------------------------------------------------------------------------------------------------------------------------------------------------------------------------------------------------------------------------------

\subsubsection{\label{sec:recipe}Recipe for Valuation}
Even with the mortality derivatives, the market for insurance is incomplete when $\rho \ne \pm 1$ due to the fact that the mortality of the insured population and the mortality of the indexed population are not perfectly correlated. This mismatch is called {\it basis risk}; see \cite{CoughlanEpsteinSinha2007} for more details.  Therefore, there is no unique method of pricing for insurance contracts, and to value contracts in this market, one has to choose a pricing mechanism. For example,  \cite{BayraktarLudkovski2009Relative} used indifference pricing and \cite{DahlMoller2006Valuation} consider the set of equivalent martingale measures when pricing the unhedgeable mortality risk.  We use the instantaneous Sharpe ratio proposed by \cite{MilevskyPromislowYoung2005Financial} and \cite{BayraktarMilevskyPromislow2008Valuation} to price the risk due to the desirable properties of the resulting price.  We will show these properties in Section \ref{sec:prop_Pn}.  Moreover, as the number of contracts approaches infinity, the limiting price per contract can be represented as an expectation with respect to an equivalent martingale measure. In a market without mortality derivatives, this pricing methodology has been proved useful for pricing pure endowments (\cite{MilevskyPromislowYoung2005Financial}), life insurance (\cite{Young2008Pricing}), life annuities (\cite{BayraktarMilevskyPromislow2008Valuation}), and financial derivatives (\cite{BayraktarYoung2007Pricing}). In this paper, we extend this pricing mechanism to incorporate mortality derivatives in the financial market.  Our method for pricing in an incomplete market with mortality derivatives is as follows:
\begin{enumerate}
\item First, we set up a portfolio composed of two parts: (1) the obligation to underwrite the pure endowment, and (2) a self-financing sub-portfolio of $T$-bonds, $q$-forwards maturing at $T$, and money market funds to partially hedge the pure endowment contract. 
\item Second, we find the optimal investments in bonds and mortality derivatives to minimize the local variance of the portfolio. This method is called local risk minimization by \cite{Schweizer2001A-Guided}. In case of a complete market, the minimized local volatility is zero. However, the incompleteness of the insurance market leads to residual risk, as measured by the local variance. 
\item Third, we assume that the insurance provider requires compensations for the unhedgeable risk. The price of the contingent claim is set to make the instantaneous Sharpe ratio of the total portfolio equal to a pre-specified value. This is equivalent to setting the price of the contingent claim such that the drift of the portfolio equals the short rate times the portfolio value plus the pre-specified Sharpe ratio times the local standard deviation of the portfolio.  Thus, our pricing method is a type of local standard deviation premium principle, \cite{Young2004}.
\end{enumerate}
%------------------------------------------------------------------------------------------------------------------------------------------------------
\subsubsection{\label{sec:P}Hedging and Pricing a Single Pure Endowment}
Denote by $P(r, \li, \lp, t; T)$ the time-$t$ value of a pure endowment that pays \$1 at maturity $T$  if the individual is alive at that time.  Here, we explicitly recognize that the price of the pure endowment depends on the short rate $r$, the hazard rate $\lp$ of the insured individual, and the hazard rate $\li$ of the indexed population. Since the maturity $T$ is fixed, we simplify the notation to $P(r, \li, \lp, t)$. (By writing $P$ to represent the value of the pure endowment, we assume that the individual is alive. If the individual dies before $T$, the value of the pure endowment jumps to \$0.) 

Suppose the insurer creates a portfolio $\Pi$ as described in Step (i) in Section \ref{sec:recipe}.  This portfolio is consist of two parts: (1) the obligation to underwrite the pure endowment with value $-P$, and (2) a self-financing sub-portfolio $V_t$ of $T$-bonds,  $q$-forwards,  and money market funds to hedge the risk of the pure endowment. Thus, $\Pi_t = -P(r_t, \li_t, \lp_t, t) + V_t$.  Let $\pi^r_t$ equal the number $T$-bonds and $\pi^{\li}_t$ the number of $q$-forwards in the self-financing sub-portfolio at time $t$ with the rest, namely, $V_t -\pi^r_t F(r_t, t) - \pi^{\li}_t S(r_t, \li_t, t)$, in money market funds. 

By It\^o's lemma, the dynamics of the value of the pure endowment $P(r, \li, \lp, t)$ in the {\it physical} probability space is  given by 
\begin{equation}\label{dyn_P}
\begin{split}
\d P( r, \li, \lp,t) &= P_t \, \d t + P_r \, \d r_t + P_{\li} \, \d \li_t  + P_{\lp} \, \d \lp_t + \dfrac{1}{2} P_{rr} \, \d[r,r]_t \\
&\quad + \dfrac{1}{2} P_{\li \lp} \, \d[\li,\li]_t + P_{\li \lp} \, \d[\li, \lp]_t + \dfrac{1}{2} P_{\lp \lp} \, \d[\lp,\lp]_t  - P \, \d N_t\\
 &= \left[ P_t  +  \mu \, P_r + a^I \cdot \left( \li_t - \uli \right) P_{\li} + a^P \cdot \left( \lp_t - \ulp \right) P_{\lp}  \right] \d t \\
 & \quad + \left[ \dfrac{1}{2} \, \sigma^2 \, P_{rr} + \frac{1}{2}  \left( b^I \right)^2 \left( \li_t - \uli \right)^2 P_{\li \li} + \dfrac{1}{2} \left( b^P \right)^2 \left( \lp_t - \ulp \right)^2 P_{\lp \lp}  \right] \d t\\
   & \quad + \rho\ b^I b^P \left( \li_t - \uli \right) \left( \lp_t - \ulp \right)  P_{\li \lp} \, \d t -  P \, \d N_t\\
   & \quad +\sigma \, P_r \, \d W_t^{r} + b^I \cdot \left( \li_t - \uli \right) P_{\li} \, \d W^I_t +  b^P \cdot \left(\l^P_t - \ulp \right) P_{\lp} \, \d W^P_t ,
\end{split}
\end{equation}
in which $[\cdot,\cdot]_t$ represents the quadratic variation at time $t$, and $N_t$ is a time-inhomogeneous Poisson process with intensity $\lp_t$  that indicates when the individual dies. Recall that the value of $P$ jumps to \$0 when the individual dies; thus, we have the term $- P \, \d N_t$ to account for this drop.

Since the sub-portfolio $V_t$ is self-financing, its dynamics are given by
\begin{equation}\label{dyn_V}
\begin{split}
\d V_t &= \pi_t^r \, \d F(r_t, t) + \pi^{\li}_t \d S_t\left(r_t, \li_t, t \right) + r_t\left[V_t - \pi_t^r \, F(r_t, t ) - \pi^{\li}_t S_t \left(r_t, \li_t, t \right)\right]\d t\\
           &=  \left[ \pi_t^r \, q^r_t \, \sigma \, F_r  + \pi^{\li}_t q^r_t \, \sigma \, S_r + \pi^{\li}_t q^{\li}_t \, b^I \cdot \left(\l^I_t-\uli \right)S_{\li} + r_t \, V_t \right]\d t\\
           & \quad + \left( \pi_t^r \, \sigma \, F_r + \pi^{\li}_t \sigma \, S_r \right) \d W_t^r +  \pi^{\li}_t b^I \cdot \left(\l^I_t-\uli \right) S_{\li} \, \d W^I_t,\\
\end{split}
\end{equation}
in which the second equality follows from equations \eqref{eqn_F} and \eqref{FK_S}, and we suppress the dependence of the functions on the underlying variables. 

It follows from equations \eqref{dyn_P} and \eqref{dyn_V} that the value of the portfolio $\Pi_{t+h}$ at time $t+h$ for $h > 0$, given $\Pi_t = \Pi$, is 
\begin{equation}\label{Pi_h}
\begin{split}
\Pi_{t+h} &= \Pi - \int_t^{t+h} \d P(r_s, \lp_s , \li_s, s) +  \int_t^{t+h} \d V_s\\
                &= \Pi -  \int_t^{t+h} \D P(r_s, \lp_s , \li_s, s) \, \d s \\
                & \quad + \int_t^{t+h}  \left[ \pi_s^r q_s^r \sigma F_r  + \pi^{\li}_s q_s^r \sigma S_r + \pi^{\li}_s q_s^{\li} b^I \cdot \left(\l^I_s-\uli \right)S_{\li}  \right]\d s\\\
                & \quad + \int_t^{t+h}  \left( \pi_s^r \sigma  F_r + \pi^{\li}_s \sigma S_r-\sigma P_r \right) \d W_s^{r} + \int_t^{t+h} b^I \cdot \left(\l^I_s-\uli \right)  \left(\pi^{\li}_s S_{\li} - P_{\li} \right)\d W_s^{I} \\
                & \quad - \int_t^{t+h} b^P \cdot \left(\l^P_s - \ulp \right) P_{\lp} \, \d W_s^{P}  +   \int_t^{t+h}  P \, (\d N_s - \lp_s \, ds) + \int_t^{t+h}r_s \, \Pi_s \, \d s,
\end{split}
\end{equation}
in which $\D$ is the operator defined on the set of appropriately differentiable functions on  $\R_{+}\times(\uli, \infty) \times (\ulp, \infty) \times [0,T]$ by 
\begin{equation}
\label{dp}
\begin{split}
\D v= &- \left(r + \lp \right) v + v_t + \mu v_r +a^I \cdot \lia v_{\li} + a^P \cdot \lpa v_{\lp}  + \dfrac{1}{2} \sigma^2 v_{rr} \\
                 &+ \frac{1}{2} \left(b^I \right)^2 \lia^2v_{\li\li}  +  \rho \, b^I b^P \lia \lpa v_{\li \lp} +\dfrac{1}{2} \left(b^P \right)^2 \lpa^2 v_{\lp\lp} .\\
\end{split}
\end{equation}
Note that the compensated counting process $N_t - \int_0^t \lp_s \, \d s $ is a (local) martingale . 

When we consider the single life case, the value of $P_t$ becomes zero immediately after the individual's death, so the value of the portfolio increase by $P$.  If we consider the price $\Pn$ of n conditionally independent and identically distributed lives as in Section  \ref{sec:Pn}, the intensity of the counting process ${N_t}$ is $n\lp_t$ at time $t$. As one of the $n$ individual dies, the value of the portfolio increases by $\Pn-P^{(n-1)}$.  We will consider $\Pn$ later and continue with the single-life case for now. 

The second step as stated in Section \ref{sec:recipe} is to choose $\pi^r_t$ and $\pi^{\li}_t$ to minimize the local variance of the portfolio.  To this end, we calculate the conditional expectation and variance of $\Pi_{t+h}$ at time $t$ given $\Pi_t=\Pi$.  First, we define a stochastic process $Y_h$ for $h>0$ by
\begin{equation}\label{Y_h}
\begin{split}
Y_h &=  \Pi -  \int_t^{t+h} \D P(r_s, \li_s , \lp_s, s) \, \d s  + \int_t^{t+h} r_s \, \Pi_s \, \d s\\
                & \quad + \int_t^{t+h}  \left[ \pi_s^r \, q_s^r \, \sigma \, F_r  + \pi^{\li}_s \, q^r_s \, \sigma \, S_r + \pi^{\li}_s q^{\li}_s b^I \cdot \left  (\l^I_s - \uli \right) S_{\li}  \right]\d s.\\
\end{split}
\end{equation}
Thus, $\E(\Pi_{t+h} \big| \F_t) = \E^{r,\li,\lp, t}(Y_h)$, in which $ \E^{r,\li,\lp, t}$ denotes the conditional expectation given $r_t=r$, $\li_t=\li$, and $\lp_t=\lp$.  From \eqref{Pi_h} and \eqref{Y_h} we have 
\begin{equation} \label{Pi_h1}
\begin{split}
\Pi_{t+h} &= Y_h + \int_t^{t+h}  \left( \pi_s^r \sigma  F_r + \pi^{\li}_s \sigma S_r-\sigma P_r \right) \d W_s^{r} + \int_t^{t+h} b^I \cdot \left(\l^I_s-\uli \right)  \left(\pi^{\li}_s S_{\li} - P_{\li} \right)\d W_s^{I} \\
                & \quad - \int_t^{t+h} b^P \cdot \left(\l^P_s - \ulp \right) P_{\lp} \, \d W_s^{P} +   \int_t^{t+h}  P \, (\d N_s - \lp_s \, ds).
\end{split}
\end{equation}
It follows that 
\begin{equation}\label{var_Pi_h}
\begin{split}
 & \mathrm{Var} \left[ \Pi_{t+h}\big| \F_t\right] = \E \left[\left(\Pi_{t+h}- \E^{r,\li,\lp,t}(Y_h)\right)^2\bigg| \F_t\right]  \\
  &\quad = \E(Y_h-\E Y_h)^2 + \E \int_t^{t+h}  \left( \pi_s^r \sigma  F_r + \pi^{\li}_s \sigma S_r-\sigma P_r \right)^2 \d s\\	
  & \qquad + \E \int_t^{t+h} \left(b^I\right)^2\left(\l^I_s-\uli \right)^2  \left(\pi^{\li}_s S_{\li} - P_{\li} \right)^2 \d s + \E \int_t^{t+h} \left(b^P\right)^2 \left(\l^P_s - \ulp \right)^2 \left(P_{\lp}\right)^2 \d s  \\
  & \qquad - 2 \, \E \int_t^{t+h}  \rho \, b^I b^P \left(\l^I_s-\uli \right) \left(\l^P_s-\ulp \right) P_{\lp}  \left(\pi^{\li}_s S_{\li} - P_{\li} \right) \d s  + \E \int_t^{t+h}  \lp P^2 ds+o(h), \\
\end{split}
\end{equation}
in which all the expectations are conditional on the information available at time $t$. Thus, the optimal investments in the $q$-forward and $T$-bond to minimize the local variance are given by, respectively, 
\begin{eqnarray}
\left( {\pi^{\li}_t} \right)^{\ast} &=& \frac{1}{S_{\li}} \left [ P_{\li} +\rho \; \frac{ b^P \cdot \left( \lp_t - \ulp \right)}{b^I \cdot \left( \li_t - \uli \right)} \, P_{\lp} \right], \label{pilistar} \\
\left( {\pi^r_t} \right)^{\ast}&=&  \frac{1}{F_r} \left(P_r - \left( {\pili_t} \right)^{\ast} S_r \right)  \label{pirstar}.
\end{eqnarray}
Equations \eqref{pilistar} and \eqref{pirstar} show that in the self-financing sub-portfolio, the $q$-forward is used to hedge the mortality risk in the pure endowment, and $T$-bonds are used to hedge the interest risk of the portfolio.  Under this investment strategy, the drift and local variance of the portfolio become
\begin{equation}\label{local_drift}
\lim_{h\to 0}\dfrac{1}{h} \left[\E(\Pi_{t+h} \big| \F_t) -\Pi \right] = -\D^Q P + r\Pi, 
\end{equation}
and 
\begin{equation}\label{local_var}
\lim_{h\to 0}\dfrac{1}{h} \mathrm{Var} \left[\Pi_{t+h} \big| \F_t\right] = (1-\rho^2)\left(b^P\right)^2 \left(\lp-\ulp\right)^2 P^2_{\lp}+\lp P^2, 
\end{equation}
with 
\begin{equation}\label{DPQ}
\begin{split}
\D^Q P &= - \left( r+\lp \right) P + P_t+\mu^Q P_r + a^{I, Q} \lia  P_{\li} + a^{P, Q} \lpa P_{\lp}  + \frac{1}{2}\sigma^2 P_{rr} \\
& \quad + \frac{1}{2} \left(b^I \right)^2 \lia^2 P_{\li\li} +  \rho \, b^I b^P \lia \lpa P_{\li \lp}   + \dfrac{1}{2} \left(b^P \right)^2 \lpa^2 P_{\lp\lp}.
\end{split}
\end{equation}
\begin{remark} \label{rmk:21} 
When $\rho=\pm1$, the $q$-forward and the pure endowment bear identical uncertainty risk in the hazard rates. In this case, the mortality risk in the pure endowment can be completely hedged with the $q$-forward, and the minimum local variance of the portfolio only comes from the random occurrence of death, namely,  
\begin{equation} \label{eq:227}
\lim_{h\to 0}\dfrac{1}{h} \mathrm{Var} \left[\Pi_{t+h} \bigg| \F_t\right] = \lp P^2.
\end{equation}
%in which $\Pi^{(n)}$ represents a portfolio with $n$ pure endowments. 
\end{remark}

\begin{remark}\label{rmk_pil0}
As we will show in Property \ref{ppt_rho0} in Section \ref{sec:prop_Pn}, $P_{\li}\equiv0$ when $\rho=0$, and  the corresponding optimal investment in the $q$-forward is ${\pi^{\li}_t}^{\ast} \equiv 0$.  Intuitively, the $q$-forward is not used to hedge the mortality risk in the pure endowment when the two underlying hazard rates are not correlated.
\end{remark}

Next, we price the pure endowment via the instantaneous Sharpe ratio as stated in Step (iii) in Section \ref{sec:recipe}. The minimized local variance of the portfolio in \eqref{local_var} is positive; therefore, the insurer is not able to hedge all the risk underlying the pure endowment. The insurer requires an excess return on this unhedgeable risk so that the instantaneous Sharpe ratio of the portfolio equals a pre-specified value $\alpha$.  We could allow $\alpha$ to be a function of say $r$, $\li$, $\lp$, and t to parallel the market price of the risk process $\left\{q^r_t, q_t^{\li}\right\}$. However, for simplicity we choose $\alpha$ to be a constant.  (Further discussion of the instantaneous Sharpe ratio is available in \cite{MilevskyPromislowYoung2006Killing}.)   We assume that $0\le \alpha \le \sqrt{\ulp}$; as we will see, some of the properties of $P$ rely on this upper bound for $\alpha$.

To achieve a Sharpe ratio of $\alpha$ and thereby to determine the value $P$ of the pure endowment, we set the drift of the portfolio equal to short rate times the portfolio value plus $\alpha$ times the minimized local standard deviation of the portfolio. Thus, we get the following equation for $P$ from \eqref{local_drift} and \eqref{local_var}:
\begin{equation}
-\D^Q P + r\Pi= r\Pi + \alpha \sqrt{ (1-\rho^2)\left(b^P\right)^2 \left(\lp-\ulp\right)^2 P^2_{\lp}+\lp P^2}. 
\end{equation}
 If the individual is still alive at time $T$, then the policy is worth exactly \$1 at that time, that is,  $P\left(r, \li, \lp, T\right)=1$. Thus, $P=P(r, \li, \lp,t)$ solves the following non-linear PDE  on $\R_{+} \times (\uli, \infty) \times (\ulp, \infty) \times [0,T]$:
\begin{equation}\label{eqn_P}
\begin{cases}
P_t+\mu^Q P_r +a^{I,Q} \cdot \lia  P_{\li} + a^{P,Q} \cdot \lpa P_{\lp}    \\
\quad + \frac{1}{2}\sigma^2 P_{rr} + \dfrac{1}{2}\left(b^I \right)^2 \lia^2 P_{\li \li} + \dfrac{1}{2} \left( b^P \right)^2 \lpa^2 P_{\lp \lp}  \\
\quad+ \rho \, b^I b^P \lia \lpa P_{\li \lp} - \left( r+\lp \right) P\\
\quad = -\alpha \sqrt{\left(1-\rho^2 \right) \left(b^P\right)^2 \lpa^2 P_{\lp}^2+\lp P^2},\\
P(r,\li,\lp,T)=1.
\end{cases}
\end{equation}

We can simplify the solution to \eqref{eqn_P} because the uncertainty  in the short rate is uncorrelated with the uncertainty in mortality rates.  Indeed, note that $P(r, \li, \lp,t) = F(r, t) \, \psi\left(\li, \lp, t \right)$, in which $F$ is the price of the $T$-bond and solves \eqref{FK_F}, and $\psi$ solves the following non-linear PDE:
\begin{equation}\label{eqn_psi}
\begin{cases}
\psi_t  + a^{I,Q} \cdot \lia \psi_{\li} + a^{P,Q} \cdot \lpa  \psi_{\lp}  + \dfrac{1}{2}\left(b^I \right)^2 \lia^2 \psi_{\li \li} \\
\quad + \rho \, b^I b^P \lia \lpa \psi_{\li \lp} + \dfrac{1}{2} \left( b^P \right)^2 \lpa^2 \psi_{\lp \lp}  - \lp \psi \\
\quad = -\alpha \sqrt{\left(1-\rho^2 \right) \left(b^P\right)^2 \lpa^2 \psi_{\lp}^2+\lp \psi^2},\\
\psi\left(\li, \lp, T \right) = 1.
\end{cases}
\end{equation}
The existence of a solution to \eqref{eqn_psi} follows from standard techniques; see, for example, Chapter 36 in \cite{Walter}. Uniqueness of the solution follows from the comparison result in Section \ref{sec:prop_Pn} of this paper.

%-----------------------------------------------------------------------------------------------------------------------------------------------------------------------------------------------------------------------------------------------------
\subsubsection{\label{sec:Pn}Hedging and Pricing a Portfolio of Pure Endowments}

In this section, we develop the PDE for the price $\Pn$ of  $n$  pure endowment contracts.  We assume that all the individuals are of the same age and are subject to the same hazard rate given in \eqref{lp}. We further assume that, given the hazard rate,  occurrences of death are  independent.  As discussed in the paragraph following equation \eqref{dp}, when an individual dies, the portfolio value $\Pi$ increases by $\Pn-P^{(n-1)}$. By paralleling the derivation of \eqref{eqn_P}, one gets the following PDE for $\Pn$: 
\begin{equation}\label{eqn_Pn}
\begin{cases}
\Pn_t + \mu^Q \Pn_r + a^{I,Q} \cdot \lia \Pn_{\li} + a^{P,Q} \cdot \lpa  \Pn_{\lp}  \\ 
\quad + \frac{1}{2}\sigma^2 \Pn_{rr}  + \dfrac{1}{2}\left(b^I \right)^2 \lia^2 \Pn_{\li \li} + \dfrac{1}{2} \left( b^P \right)^2 \lpa^2 \Pn_{\lp \lp} \\
\quad + \rho \, b^I b^P \lia \lpa \Pn_{\li \lp} - r \Pn - n \lp \cdot \left(\Pn-P^{(n-1)}\right)\\
\quad = -\alpha \sqrt{\left(1-\rho^2 \right) \left(b^P\right)^2 \lpa^2 \left(\Pn_{\lp}\right)^2+n \lp \left (\Pn-P^{(n-1)}\right)^2},\\
P^{(n)}(r,\li,\lo,T)=n,
\end{cases}
\end{equation}
with initial value $P^{(0)}\equiv0$, and $P^{(1)} = P$, as given by \eqref{eqn_P}. 

As in Section \ref{sec:P}, $\Pn(r, \li, \lp, t) = F(r, t) \, \psin (\li, \lp, t)$, in which $F$ solves \eqref{FK_F} and $\psin$ solves the following PDE
\begin{equation}\label{eqn_psin}
\begin{cases}
\psin_t + a^{I,Q} \cdot \lia \psin_{\li} + a^{P,Q} \cdot \lpa  \psin_{\lp} + \dfrac{1}{2} \left( b^I \right)^2 \lia^2 \psin_{\li \li}  \\
\quad + \rho \, b^I b^P \lia \lpa \psin_{\li \lp} + \dfrac{1}{2}\left(b^P \right)^2 \lpa^2 \psin_{\lp \lp}   -  n \lp \cdot \left(\psin-\psi^{(n-1)}\right)\\
\quad = -\alpha \sqrt{\left(1-\rho^2 \right) \left(b^P\right)^2 \lpa^2 \left(\psin_{\lp}\right)^2+ n \lp \left(\psin-\psi^{(n-1)}\right)^2},\\
\psin\left(\li, \lp, T\right)=n, 
\end{cases}
\end{equation}
with initial value $\psi^{(0)}\equiv0$, and $\psi^{(1)}= \psi$, as given by \eqref{eqn_psi}.

%-----------------------------------------------------------------------------------------------------------------------------------------------------------------------------------------------------------------------------------------------------
\section{\label{sec:prop_Pn}properties of $P^{(n)}$}
To demonstrate properties of $\Pn$, we need a comparison principle similar to the one in \cite{Walter}. To this end,  we first state a relevant one-sided Lipschitz condition along with growth conditions. We require that  the function $g=g\left( \li,\lp,t,v,p_1,p_2 \right)$ satisfies the following one-sided Lipschitz condition:  For $v>w$,
\begin{equation}\label{Lipschitz}
\begin{split} 
g\left( \li,\lp,t,v,p_1,p_2\right)-g\left( \li,\lp,t,w,q_1,q_2 \right) & \le  c\left(\li, \lp, t\right)(v-w)+d_1\left(\li, \lp, t\right)|p_1-q_1| \\ 
&\quad+d_2\left(\li, \lp, t\right)|p_2-q_2|,
\end{split}
\end{equation}
with growth conditions on $c$, $d_1$ and $d_2$ given by 
\begin{equation}\label{growth}
\begin{cases}
0\le c\left(\li, \lp, t\right) \le K\left[1+\left(\ln \lia\right)^2+\left(\ln \lpa\right)^2\right],\\
0\le d_1\left(\li, \lp, t\right) \le K\lia\left[1+\ln \lia+\ln \lpa\right],\\
0\le d_2\left(\li, \lp, t\right) \le K\lpa\left[1+\ln \lia+\ln \lpa\right],\\
\end{cases}
\end{equation}
for some constant $K \ge 0$ and for all $\left(\li, \lp, t\right)\in (\uli, \infty) \times (\ulp, \infty) \times [0,T]$.

To prove Lemma \ref{lem_gn} below, as well as many of the properties of $\Pn$, we rely on the following lemma.
\begin{lemm}\label{lem:ABC}
$\dd{\sqrt{C^2+A^2}\le|A-B|+\sqrt{C^2+B^2}}$
\end{lemm}
\begin{proof}
It is clear that the inequality holds if $A\le B$. For the case $A>B$, see the proof of Lemma $4.5$ in \cite{MilevskyPromislowYoung2005Financial}. 
\end{proof}

\begin{lemm}\label{lem_gn}
Define $g_n$, for $n\ge1$, by
\begin{equation}\label{eqn_gn}
\begin{split}
g_n\left( \li,\lp,t,v,p_1,p_2\right) &=a^{I,Q} \cdot \left(\li -\uli \right)p_1+a^{P,Q} \cdot \left(\lp-\ulp\right)p_2 -n\lp \left(v-\psi^{(n-1)}\right)  \\
& \quad+\alpha \sqrt{\left(1-\rho^2 \right) \left(b^P\right)^2 \lpa^2 p_2^2+ n \lp \left(v-\psi^{(n-1)}\right)^2},\\
\end{split}
\end{equation}
in which $\psi^{(n-1)}$ solves \eqref{eqn_psin} with $n$ replaced by $n-1$. Then, $g_n$ satisfies the one-sided Lipschitz condition \eqref{Lipschitz} on $\left(\li, \lp, t\right)\in (\uli, \infty) \times (\ulp, \infty) \times [0,T]$. Furthermore,  condition \eqref{growth} holds if 
\begin{equation}\label{assumption1}
\begin{cases}
\big| a^{I,Q} \big| \leq  K\left[1+\ln \lia + \ln \lpa \right], \\
\big|a^{P,Q} \big| \leq K\left[1+\ln \lia + \ln \lpa \right], \\
\end{cases}
\end{equation}
for some constant $K \ge 0$.
\end{lemm}

\begin{proof}
Suppose that $v>w$, then 
\begin{equation}\label{cond_gn}
\begin{split}
&  \ g_n\left( \li,\lp,t,v,p_1,p_2\right)-g_n\left( \li,\lp,t,w,q_1,q_2\right)\\
& \quad = a^{I,Q}\lia(p_1-q_1)+a^{P,Q} \lpa (p_2-q_2)- n \lp(v-w)\\
&\qquad+\alpha \sqrt{\left(1-\rho^2 \right) \left(b^P\right)^2 \lpa^2p_2^2 + n \lp \left(v-\psi^{(n-1)}\right)^2}\\
& \qquad - \alpha \sqrt{\left(1-\rho^2 \right) \left(b^P\right)^2 \lpa^2q_2^2 + n \lp \left(w-\psi^{(n-1)}\right)^2}\\
& \quad \le \big| a^{I,Q}\big|\lia|p_1-q_1| + \left[ \big| a^{P,Q} \big| \lpa+ \alpha\sqrt{1-\rho^2} \, b^P\lpa \right]|p_2-q_2|\\
&\qquad -\left(n \lp-\alpha\sqrt{n \lp}\right)(v-w)\\
&\quad \le \big|a^{I,Q}\big| \lia |p_1-q_1| + \left[ \big|a^{P,Q} \big| + \alpha \sqrt{1-\rho^2} \, b^P  \right] \lpa |p_2-q_2|. 
\end{split}
\end{equation}
In the above series of inequalities, we use $\alpha\le\sqrt{\underline{\l}^P}\le\sqrt{{\lp}}$ and Lemma \ref{lem:ABC}. Therefore, \eqref{Lipschitz} holds with $c = 0$, $d_1 = \left|a^{I,Q}\right|\lia$ and $d_2 = \left|a^{P,Q}  \right|\lpa+ \alpha\sqrt{1-\rho^2} \, b^P\lpa$.  Notice that $d_1$ and $d_2$ satisfy condition \eqref{growth} if \eqref{assumption1} holds.
\end{proof}

\begin{assumption}\label{assum_growth}
Henceforth,  we assume that the condition \eqref{assumption1} holds for rest of the paper.  For later purpose, we also assume that $a^{P,Q}_{\lp} \lpa$ is H\"older continuous and satisfies the following growth condition
\begin{equation}\label{growth2}
\left| a^{P,Q}_{\lp} \lpa + a^{P,Q} \right| \le K\left[1+\left(\ln\lpa\right)^2\right].
\end{equation}
\end{assumption}

\begin{thm}\label{thm_comparison}
Let $G= (\uli, \infty) \times (\ulp, \infty) \times [0,T]$, and denote by $\mathcal{G}$ the collection of functions on $G$ that are twice differentiable in their first two variables and once-differentiable in their third variable. Define an operator $\L$ on $\mathcal{G}$ by
\begin{equation}\label{operatorL}
\begin{split}
\L v & =v_t + \dfrac{1}{2}\left(b^I \right)^2 \lia^2 v_{\li \li} + \rho \, b^I \, b^P \lia \lpa \psi_{\li \lp} \\
 &\quad  + \dfrac{1}{2} \left( b^P \right)^2 \lpa^2 v_{\lp \lp} + g_n \left( \li,\lp,t,v,v_{\li},v_{\lp} \right),\\
\end{split}
\end{equation}
in which $g_n$ is given by \eqref{eqn_gn}.  Suppose that $v,w \in \mathcal{G}$ are such that there exists a constant $K \ge 0$ with $v\le e^{K \left\{ \left(\ln \lia \right)^2 + \left(\ln \lpa \right)^2 \right\}}$ and $w\ge - e^{K\left\{\left(\ln \lia \right)^2 + \left(\ln \lpa \right)^2 \right\}}$ for large $\left(\ln \lia \right)^2 + \left(\ln \lpa \right)^2$. Then, if $(a)$ $\L v \ge \L w$ on $G$ and if $(b)$ $v\left(\li, \lp, T\right) \le w\left(\li, \lp, T\right)$ for all $\li > \uli$ and $\lp > \ulp$, then $v \le w$ on $G$. 
\end{thm}

\begin{proof}
Define $y_1=\ln\lia$, $y_2=\ln \lpa$, and $\tau=T-t$. Write $\tilde{v}(y_1,y_2,\tau)=v\left(\li, \lp, t \right)$, etc. Therefore, $ \tilde{v}\le e^{K\left(y_1^2+y_2^2\right)} $and $\tilde{w}\ge -e^{K\left(y_1^2+y_2^2\right)}$ for large $y_1^2 + y_2^2$. Under this transformation, \eqref{operatorL} becomes
\begin{equation}
\label{operatorLt}
\L\tilde{v} = - \tilde{v}_\tau + \frac{1}{2} (\widetilde{b^I})^2 \, \tilde{v}_{y_1 y_1} + \rho \, \widetilde{b^I} \, \widetilde{b^P} \, \tilde{v}_{y_1 y_2} + \frac{1}{2} (\widetilde{b^P})^2 \, \tilde{v}_{y_2 y_2} + \tilde{h}(y_1,y_2,v,\tilde{v}_{y_1} ,\tilde{v}_{y_2}),
\end{equation}
in which 
\begin{equation}\label{h}
\tilde{h}(y_1,y_2,\tau,\tilde{v},\tilde{p}_1,\tilde{p}_2) = -\dfrac{1}{2}(\widetilde{b^I})^2 \, \tilde{p}_1 - \dfrac{1}{2}(\widetilde{b^P})^2 \, \tilde{p}_2 + \tilde{g}_n(y_1,y_2,\tau,\tilde{v},\tilde{p}_1,\tilde{p}_2),
\end{equation}
and $\tilde{v}$ is a differentiable function defined on $\R^2\times[0,T]$. The differential operator in \eqref{operatorLt} is of the form considered by \cite{Walter}.  %Notice that the coefficients of second derivatives in \eqref{operatorLt} are positive  semi-definite since $\rho \in [-1, 1]$.

To complete the proof, we consider the Lipschitz and growth conditions in the original variables $\li$, $\lp$, and $t$.  From \cite{Walter}, we know that the conditions on $\tilde h$ required for Walter's comparison principle are
\begin{equation}\label{Lipschitz1}
\begin{split}
\tilde{h}(y_1,y_2,\tau,\tilde{v},\tilde{p}_1,\tilde{p}_2)-\tilde{h}(y_1,y_2,\tau,\tilde{w},\tilde{q}_1,\tilde{q}_2)&\le \tilde{c}(y_1,y_2,\tau)(\tilde{v}-\tilde{w}) +\tilde{d}_1(y_1,y_2,\tau)|\tilde{p}_1-\tilde{q}_1|\\
&\quad +\tilde{d}_2(y_1,y_2,\tau)|\tilde{p}_2-\tilde{q}_2|,\\
\end{split}
\end{equation}
with
\begin{equation}\label{growth1}
\begin{cases}
0\le \tilde{c}(y_1,y_2,\tau) \le K\left(1+y_1^2+y_2^2\right),\\
0\le \tilde{d}_1(y_1,y_2,\tau) \le K\left(1+|y_1|+|y_2|\right),\\
0\le \tilde{d}_2(y_1,y_2,\tau) \le K\left(1+|y_1|+|y_2|\right).\\
\end{cases}
\end{equation}
Under the original variables, it follows from \eqref{h} and \eqref{cond_gn} that, for $v > w$,
\begin{equation}
\begin{split}
\tilde{h}(y_1,y_2,\tau,\tilde{v},\tilde{p}_1,\tilde{p}_2) & -\tilde{h}(y_1,y_2,\tau,\tilde{w},\tilde{q}_1,\tilde{q}_2) \\
& \le \left[ \dfrac{1}{2}(\widetilde{b^I})^2 + \left| \widetilde{a^{I,Q}} \right| \right]|\tilde{p}_1-\tilde{q}_1| +\left[\dfrac{1}{2}(\widetilde{b^P})^2 + \left| \widetilde{a^{P,Q}}  \right| + \alpha\sqrt{1-\rho^2} \, \widetilde{b^P} \right]|\tilde{p}_2-\tilde{q}_2|.\\
\end{split}
\end{equation}
Note that $p_1$=$e^{-y_1}\tilde{p}_1$ since $\psi_{\li}=e^{-y}\tilde{\psi}_{y_1}$; similarly, for $\tilde{p}_2$, $\tilde{q}_1$, and $\tilde{q}_2$.  Thus, \eqref{Lipschitz1} is satisfied with $\tilde{c} = c = 0$, $\tilde{d}_1 = (\widetilde{b^I})^2 + \left| \widetilde{a^{I,Q}} \right| $, and  $\tilde{d}_2 = (\widetilde{b^P})^2 + \left| \widetilde{a^{P,Q}}  \right| + \alpha\sqrt{1-\rho^2} \, \widetilde{b^P} $, and \eqref{growth1} is satisfied due to Lemma \ref{lem_gn} and \eqref{assumption1} and the fact that $\widetilde{b^I}$ and $\widetilde{b^P}$ are continuous on $[0,T]$ and are, thus, bounded. 
\end{proof}

For the remainder of this section, we apply Theorem \ref{thm_comparison} to investigate properties of the price $\Pn$ for $n$ pure endowment contracts.  For simplicity, we will state and prove properties of $\psin$ and afterwards interpret the results in terms of $\Pn$.

\begin{ppt}\label{ppt_bound}
For $n \ge 0$, $0 \leq \psi^{(n)} \le n \, e^{-\left(\ulp-\alpha \sqrt{\ulp}\right)(T-t)}$ on $G$.
\end{ppt}
\begin{proof}
For ease of presentation, define $h$ by $h(t) = e^{-\left(\ulp-\alpha \sqrt{\ulp}\right)(T-t)}$ for $t \in [0, T]$. We proceed by induction to prove that $\psin \le n \, h$ on $G$.  Note that the inequality holds for $n=0$ since $\psi^{(0)}\equiv 0$. For $n\ge1$,  assume that $\psi^{(n-1)}(\li, \lp, t) \le (n-1)h$, and show that $0\le \psi^{(n)}(\li, \lp, t)\le n \, h$. 

To apply Theorem \ref{thm_comparison}, define a differential operator $\L$ on $\G$ by \eqref{operatorL}. We have $\L \psin=0$ due to equation \eqref{eqn_psin}.  Apply the operator $\L$ to $n \, h$ to get
\begin{equation}
\begin{split}
\L \left(n \, h \right) & = \left(\ulp - \alpha \sqrt{\ulp} \right) n \, h - \left(n \lp-\alpha \sqrt{n\lp} \right) \left(n \, h - \psi^{(n-1)}\right) \\
              & \le \left(\ulp - \alpha \sqrt{\ulp} \right) n \, h - \left(n \lp-\alpha \sqrt{n\lp} \right) \left(n-(n-1)\right)h \\
              & = \left[n\left( \ulp-\alpha \sqrt{\ulp} \right)-\left(n\lp-\alpha \sqrt{n\lp} \right)\right]h \le 0.
\end{split}
\end{equation}
Because $\L (n \, h) \le 0 = \L\psin$ and $ n \,h(T)=\psin(\li ,\lp , T)=n$, Theorem \ref{thm_comparison} implies that $\psi^{(n)} \le ne^{-\left(\ulp-\alpha \sqrt{\ulp}\right)(T-t)}$ on $G$. 

Similarly, we prove that $\psin \ge 0$ by induction. Suppose that $\psi^{(n-1)}\ge 0$ for $n\ge1$, and show that $\psin\ge 0$. We apply the same operator $\L$ from the first part of this proof to the constant function $\0$ on $G$.  Because $\L\0=\left(n\lp+\alpha \sqrt{n\lp}\right)\psi^{(n-1)}\ge 0 = \L \psin $ and $ 0 \le n = \psin(\li,\lp, T)$, Theorem \ref{thm_comparison} implies that $\psin \ge 0$ on $G$. 
\end{proof}

It follows immediately from Property \ref{ppt_bound} that $0 \le \Pn (r,\li,\lp,t) \le n \, F(r,t)$ for $ (r,\li,\lp,t) \in \R^+ \times G$, in which $F$ is the price of a $T$-bond with face value of \$$1$. Thus, the price per risk $\dfrac{1}{n}P^{(n)}$ lies between $0$ and $F$. This is a no-arbitrage condition since the total payoff of $n$ pure endowments at time $T$ is non-negative and is no more than \$$n$.

\begin{ppt}\label{ppt_mono}
For $n \ge 1$, $\psi^{(n)} \ge \psi^{(n-1)} $ on $ G$. 
\end{ppt}
\begin{proof}
We prove this property by induction. First, the inequality holds for $n=1$ since  $\psi^{(1)}\ge 0$ by Property \ref{ppt_bound} and  $\psi^{(0)}\equiv 0$.  For $n\ge 2$, assume that $\psi^{(n-1)}\ge\psi^{(n-2)}$, and show that $\psi^{(n)}\ge\psi^{(n-1)}$. 

Define a differential operator $\L$ on $\G$ by \eqref{operatorL}.  We have that $\L \psin=0$ due to equation \eqref{eqn_psin}.  Apply the operator $\L$ to $\psi^{(n-1)}$, and use the fact that $\psi^{(n-1)}$ solves $\eqref{eqn_psin}$ with $n$ replaced by $n-1$:
\begin{equation}
\begin{split}
\L\psi^{(n-1)} & = (n-1) \lp \left( \psi^{(n-1)}-\psi^{(n-2)} \right) + \alpha \sqrt{\left(1-\rho^2 \right) \left(b^P\right)^2 \lpa^2 \left(\psi^{(n-1)}_{\lp}\right)^2}\\
&\quad - \alpha \sqrt{\left(1-\rho^2 \right) \left(b^P\right)^2 \lpa^2  \left(\psi^{(n-1)}_{\lp}\right)^2 + (n-1) \lp \left(\psi^{(n-1)}-\psi^{(n-2)}\right)^2}\\
& \ge (n-1) \lp  \left(\psi^{(n-1)}-\psi^{(n-2)}\right) - \alpha \sqrt{(n-1)\lp} \left(\psi^{(n-1)}-\psi^{(n-2)}\right)\\
&= \left[(n-1)\lp-\alpha\sqrt{(n-1)\lp}\right] \left(\psi^{(n-1)}-\psi^{(n-2)}\right) \ge 0. 
\end{split}
\end{equation}
Note that the first inequality is due to the fact that $\sqrt{A^2+B^2}\le |A|+|B|$. We also use the induction assumption that $\psi^{(n-1)}\ge \psi^{(n-1)}$.  Because $\psin=n > n-1=\psi^{(n-1)}$ at $t=T$, Theorem \ref{thm_comparison} implies that $\psi^{(n)} \ge \psi^{(n-1)} $ on $G$. 
\end{proof}

We use Property \ref{ppt_mono} to prove Property \ref{ppt_wrtlp} below; however, Property \ref{ppt_mono} is interesting in its own right because it confirms our intuition that $\Pn$ increases with the number of policyholders.

\begin{ppt}\label{ppt_alpha}
Suppose $0\le \alpha_1 \le \alpha_2 \le \sqrt{\ulp}$, and let $\psi^{(n),\alpha_i}$ be the solution of \eqref{eqn_psin} with $\alpha=\alpha_i$, for $i=1,2$ and  for $n\ge0$. Then, $\psi^{(n),\alpha_1} \le \psi^{(n),\alpha_2}$ on $G$. 
\end{ppt}
\begin{proof}
We prove this property by induction.  First, the inequality holds for $n=0$ since $ \psi^{(0),\alpha_i}\equiv 0$ for $i=1,2$.  For  $n\ge 1$, assume that  $\psi^{(n-1),\alpha_1} \le \psi^{(n-2),\alpha_2}$, and show that  $\psi^{(n),\alpha_1} \le \psi^{(n),\alpha_2}$. 

Define a differential operator $\L$ on $\G$ by \eqref{operatorL} with $\alpha=\alpha_1$. We have that $\L \psi^{(n),\alpha_1}=0$ since $\psi^{(n),\alpha_1}$ solves \eqref{eqn_psin}  with $\alpha=\alpha_1$.  Apply the operator $\L$ to $\psi^{(n),\alpha_2}$ to get
\begin{equation}
\begin{split}
\L\psi^{(n),\alpha_2} & = -n\lp\left( \psi^{(n),\alpha_2}-\psi^{(n-1),\alpha_1} \right) + n\lp \left( \psi^{(n),\alpha_2}-\psi^{(n-1),\alpha_2} \right)\\
&\quad + \alpha_1 \sqrt{\left(1-\rho^2 \right) \left(b^P\right)^2 \lpa^2 \left(\psi^{(n),\alpha_2}_{\lp}\right)^2 + n\lp\left( \psi^{(n),\alpha_2}-\psi^{(n-1),\alpha_1} \right)^2}\\
&\quad - \alpha_2 \sqrt{\left(1-\rho^2 \right) \left(b^P\right)^2 \lpa^2 \left(\psi^{(n),\alpha_2}_{\lp}\right)^2 + n\lp\left( \psi^{(n),\alpha_2}-\psi^{(n-1),\alpha_2} \right)^2}\\
& = -n\lp \left( \psi^{(n-1),\alpha_2}-\psi^{(n-1),\alpha_1} \right) \\
& \quad + \alpha_1  \Bigg\{\sqrt{\left(1-\rho^2 \right) \left(b^P\right)^2 \lpa^2 \left(\psi^{(n),\alpha_2}_{\lp}\right)^2 + n\lp\left( \psi^{(n),\alpha_2}-\psi^{(n-1),\alpha_1} \right)^2}\\
&\qquad \qquad -  \sqrt{\left(1-\rho^2 \right) \left(b^P\right)^2 \lpa^2 \left(\psi^{(n),\alpha_2}_{\lp}\right)^2 + n\lp\left( \psi^{(n),\alpha_2}-\psi^{(n-1),\alpha_2} \right)^2} \Bigg\}\\
& \quad - \left(\alpha_2-\alpha_1\right) \sqrt{\left(1-\rho^2 \right) \left(b^P\right)^2 \lpa^2 \left(\psi^{(n),\alpha_2}_{\lp}\right)^2 + n\lp\left( \psi^{(n),\alpha_2}-\psi^{(n-1),\alpha_2}\right)^2}\\
& \le -\left(n\lp - \alpha_1\sqrt{n\lp} \right)\left(\psi^{(n-1),\alpha_2}-\psi^{(n-1),\alpha_1}\right) \le 0 = \L \psi^{(n),\alpha_1}. \\
\end{split}
\end{equation}
Here, we use the Lemma \ref{lem:ABC} with $A=\sqrt{n\lp}\left(\psi^{(n), \alpha_2}-\psi^{(n-1), \alpha_1}\right)$,  $B=\sqrt{n\lp}\left(\psi^{(n), \alpha_2}-\psi^{(n-1), \alpha_2}\right)$, and $C=b^P\left(\lp-\ulp\right)\psi^{(n),\alpha_2}_{\lp}$, as well as the induction hypothesis and $\alpha_2\ge\alpha_1$. Because $\psi^{(n),\alpha_1} = \psi^{(n),\alpha_1} = n$ at $t=T$,  Theorem \ref{thm_comparison} implies that $\psi^{(n),\alpha_1} \le \psi^{(n),\alpha_2}$ on $G$.  
\end{proof}

Property \ref{ppt_alpha} shows that $\Pn$ increases with the instantaneous Sharpe ratio $\alpha$.  The more that the insurance company wants to be compensated for the unhedgeable portion of the mortality risk, the higher it will set $\alpha$.  We have the following corollary of Property \ref{ppt_alpha}. 

\begin{ppt}\label{ppt_alpha0}
Let $\psi^{(n),\alpha0}$ be the solution to \eqref{eqn_psin} with $\alpha=0$. Then, for $0\le \alpha \le \sqrt{\ulp}$, $\psi^{(n),\alpha} \ge \psi^{(n),\alpha0}$ on $G$, and we can express the lower bound $\psi^{(n), \alpha0}$ as follows: $\psi^{(n), \alpha0} = n \, \psi^{\alpha0}$, in which $\psi^{\alpha0}$ is given by
\begin{equation}\label{eqn_psi0}
\psi^{\alpha0}(\li, \lp, t)=\E^{Q} \left[e^{-\int_t^T \lp_s \d s} \, \bigg| \, \li_t = \li, \lp_t = \lp \right],
\end{equation} 
and the $Q$-dynamics of $\left\{ \li_t \right\}$ and $\left\{ \lp_t \right\}$ follow, respectively,
\begin{equation} \label{li_Q}
\d \li_t = a^{I, Q}(\li_t, t)\left(\li_t-\uli \right)\d t+b^I(t)\left(\li_t-\uli \right)\d W^{I, Q}_t, 
\end{equation}
and
\begin{equation} \label{lp_Q}
\d \l^P_t = a^{P, Q}(\li_t, \l^P_t, t)\left(\l^P_t-\ulp \right)\d t+b^P(t)\left(\l^P_t-\ulp \right)\d W^{P, Q}_t.
\end{equation}
\end{ppt}

\begin{proof}
Let $\alpha_1 = 0$ and $\alpha_2 = \alpha \ge 0$ in Property \eqref{ppt_alpha}, and the inequality follows. By substituting $\alpha = 0$ in \eqref{eqn_psi}, the Feyman-Kac Theorem leads to the  expression of $\psi^{\alpha0}$ in \eqref{eqn_psi0}.    Finally, it is straightforward to show that $n \psi^{\alpha0}$ solves \eqref{eqn_psin} with $\alpha = 0$; thus, $\psi^{(n), \alpha0} = n \, \psi^{\alpha0}$.  
\end{proof}

Note that $n \,\psi^{(1), \alpha0} = n \, \psi^{\alpha0}$ is the expected number of survivors under the physical measure, so the lower bound of $\frac{1}{n} \Pn$ (as $\alpha$ approaches zero) is the same as the lower bound of $P$, namely, $F \, \psi^{\alpha0}$.

\begin{ppt}\label{ppt_wrtlp}
$\psi^{(n)}_{\lp} \le 0$ on $G$ for $n\ge0$. 
\end{ppt}

\begin{proof}
We prove this property by induction.  First, it is clear  that $\psi^{(0)}_{\lp} \equiv 0$. For $n\ge1$, assume that $\psi^{(n-1)}_{\lp}\le 0$, and apply a modified version of Theorem \ref{thm_comparison} to compare $\psi^{(n)}_{\lp}\le 0$ and the constant function $\0$.  To this end, we first differentiate $\psi^{(n)}$'s equation, \eqref{eqn_psin}, with respect to $\lp$ to get an equation for $f^{(n)}=\psi^{(n)}_{\lp}$:
\begin{equation}\label{eqn_fn}
\begin{cases}
\fn_t + \left[ a_{\lp}^{P,Q} \cdot \lpa + a^{P,Q} \right]  \fn + \left[ a^{I,Q} + \rho \, b^I \, b^P \right] \lia \fn_{\li} + \left[ a^{P,Q} +\left( b^P \right)^2 \right] \lpa \fn_{\lp}\\
\quad +\dfrac{1}{2}\left(b^I \right)^2 \lia^2 \fn_{\li \li} + \rho \, b^I \, b^P \lia \lpa \fn_{\li\lp} + \dfrac{1}{2} \left( b^P \right)^2 \lpa^2 \fn_{\lp \lp} \\
\quad -  n\left(\psin-\psi^{(n-1)}\right) -n\lp\left(\fn-f^{(n-1)}\right)\\
\quad = -\alpha \dfrac{\left(1-\rho^2 \right) \left(b^P\right)^2 \lpa \left(\fn\right)^2+\left(1-\rho^2 \right) \left(b^P\right)^2 \lpa^2 \fn \fn_{\lp}}{\sqrt{\left(1-\rho^2 \right) \left(b^P\right)^2 \lpa^2 \left(f^{(n)}\right)^2+ n \lp \left(\psin-\psi^{(n-1)}\right)^2}}\\
\quad \quad -\alpha \dfrac{\dfrac{1}{2} n \left(\psin-\psi^{(n-1)}\right)^2+n\lp   \left(\psin-\psi^{(n-1)}\right)\left(\fn-f^{(n-1)}\right)}{\sqrt{\left(1-\rho^2 \right) \left(b^P\right)^2 \lpa^2 \left(f^{(n)}\right)^2+ n \lp \left(\psin-\psi^{(n-1)}\right)^2}},\\
\fn\left(\li, \lp, T\right)=0.
\end{cases}
\end{equation}
Define a differential operator $\L$ on $\G$ by \eqref{operatorL} with $g_n$ replaced by 
\begin{equation}\label{gn_fn}
\begin{split}
\tilde g_n (\li, \lp, t, v, p_1,p_2) &= \left[ a_{\lp}^{P,Q} \cdot \lpa + a^{P,Q} \right] v + \left[ a^{I,Q} + \rho \, b^I \, b^P \right] \lia p_1 \\
& \quad +  \left[ a^{P,Q} +\left( b^P \right)^2 \right] \lpa p_2 -  n\left(\psin-\psi^{(n-1)}\right)-n\lp\left(v-f^{(n-1)}\right)\\
& \quad + \alpha \dfrac{\left(1-\rho^2 \right) \left(b^P\right)^2 \lpa v^2+\left(1-\rho^2 \right) \left(b^P\right)^2 \lpa^2 v p_2}{\sqrt{\left(1-\rho^2 \right) \left(b^P\right)^2 \lpa^2 v^2 +  n \lp \left(\psin-\psi^{(n-1)}\right)^2}}\\
& \quad + \alpha \dfrac{\dfrac{1}{2} n \left(\psin-\psi^{(n-1)}\right)^2 + n\lp  \left(\psin-\psi^{(n-1)}\right)  \left(v-f^{(n-1)}\right)}{\sqrt{\left(1-\rho^2 \right) \left(b^P\right)^2 \lpa^2 v^2+ n \lp \left(\psin-\psi^{(n-1)}\right)^2}} .
\end{split}
\end{equation} 
From \cite{Walter}, we know that we only need to verify that \eqref{Lipschitz} holds for $v>w=0=q_1=q_2$.  It is not difficult to show that
\begin{equation} \label{eq:gn_fnLip}
\begin{split}
\tilde g_n (\li, \lp, t, & v, p_1,p_2) - \tilde g_n (\li, \lp, t, 0, 0, 0) \le \left[ \left| a^{P,Q}_{\lp} \cdot \lpa + a^{P,Q} \right| + \alpha \sqrt{1-\rho^2} \, b^P \right] v \\
& + \Big[ \left| a^{I,Q} \right| + \rho \, b^I \, b^P \Big] \lia |p_1| + \left[ \left| a^{P,Q} \right| + \left( b^P \right)^2 + \alpha \sqrt{1-\rho^2} \, b^P \right] \lpa |p_2|.
\end{split}
\end{equation}
Also, by Assumption \ref{assum_growth}, the corresponding $c = \left| a^{P,Q}_{\lp} \lpa + a^{P,Q} \right| + \alpha \sqrt{1-\rho^2} \, b^P$, $d_1 = \left| a^{I,Q} \right| + \rho \, b^I \, b^P$, and $d_2 = \left| a^{P,Q} \right| + \left( b^P \right)^2 + \alpha \sqrt{1-\rho^2} \, b^P$ in \eqref{eq:gn_fnLip} satisfy the growth conditions in \eqref{growth}. 

Note that $\L \fn =0$ on $G$. Apply the operator $\L$ to the constant function $\0$ to get  $\L\0=\left(n\lp -\alpha\sqrt{n\lp}\right) f^{(n-1)}-\left(n-\alpha/2\sqrt{n/\lp}\right)\left(\psin-\psi^{(n-1)}\right) \le 0$ by the induction assumption, by Property \ref{ppt_mono}, and by the assumption that $\lp > \ulp \ge \alpha^2$. Since $\fn(\li, \lp, T)=0$, Theorem \ref{thm_comparison} implies that $\fn=\psin_{\lp}\le 0$ on $G$. 
\end{proof}

It is intuitively pleasing that $\psin_{\lp} \le 0$ because for physical survival probabilities, if the hazard rate increases, then the probability of surviving until time $T$, and thereby paying the \$1 benefit, decreases. A related result is that $\Pn$ decreases as the risk-adjusted drift of the hazard rate, $a^{P,Q}$, increases because the hazard rate tends to increase with its drift.

\begin{ppt}\label{ppt_drift}
Suppose $a_1^{P, Q} \le a_2^{P, Q}$ on $G$, and let $\psi^{(n),a_i}$ denote the solution to \eqref{eqn_psin} with $a^{P, Q} = a^{P, Q}_i$, for $i=1,2$ and for $n \ge 0$. Then, $\psi^{(n),a_1} \ge \psi^{(n),a_2}$ on $G$.
\end{ppt}
\begin{proof}
Define a differential operator $\L$ on $\G$ by \eqref{operatorL} with $a^P=a^P_1$; then, it is clear that $\L \psi^{(n),a_1}=0$. Apply this operator $\L$ to $\psi^{(n),a_2}$ to obtain
\begin{equation}
\L \psi^{(n),a_2} = \left(a^{P,Q}_1 - a^{P,Q}_2 \right) \left(\lp-\ulp\right)\psi^{(n),a_2}_{\lp} \ge 0.
\end{equation}
Since $\psi^{(n),a_1}\left(\li, \lp, T\right) = \psi^{(n),a_2}\left(\li, \lp, T\right)=n $, Theorem \ref{thm_comparison} implies that $\psi^{(n),a_1}\ge \psi^{(n),a_2}$ on $G$. 
\end{proof}

Next, we prove the subadditivity property of $P^{(n)}$. To that end, we use Lemma $4.10$ in \cite{MilevskyPromislowYoung2005Financial}. We restate the lemma, and one can find its proof in the original paper. 

\begin{lemm}\label{lem_ABCmn}
Suppose $A\ge C\ge B$, $B_{\l}$, and $C_{\l}$ are constants; then, for non-negative integers $m$ and $n$, 
\begin{equation}
\sqrt{\left(B_{\l}+C_{\l}\right)^2+(m+n)A^2} - \sqrt{n}(A-C) \le \sqrt{B_{\l}^2+mB^2}+\sqrt{C_{\l}^2+nB^2}+\sqrt{m}(A-B). 
\end{equation}
\end{lemm}

\begin{ppt}\label{ppt_subadd} 
$\psi^{(m+n)} \le \psi^{(m)}+ \psin$ for $m, n \ge 0$.
\end{ppt}

\begin{proof}
We prove this inequality by induction on $m+n$. When $m+n=0$ or $1$, we know that $\psi^{(0)} = \psi^{(0)}+ \psi^{(0)}$ and  $\psi^{(1)} = \psi^{(1)}+ \psi^{(0)}$ since $\psi^{(0)}=0$.  For $m+n\ge2$,  suppose that $\psi^{(l+k)} \le \psi^{(l)}+ \psi^{(k)}$ for any non-negative integers $k$ and $l$ such that $k+l\le m+n-1$. We need to show that $\psi^{(m+n)} \le \psi^{(m)}+ \psin$. Define $\xi =\psi^{(m)}+ \psin$ and $\eta=\psi^{(m+n)}$ on $G$. The function $\xi$ solves the PDE given by
\begin{equation}\label{eqn_xi}
\begin{cases}
\xi_t+a^{I,Q} \cdot \lia \xi_{\li} + a^{P,Q} c\dot \lpa  \xi_{\lp} +\dfrac{1}{2}\left(b^I \right)^2 \lia^2 \xi_{\li \li} \\
\quad + \rho \, b^I \, b^P \lia \lpa \xi_{\li \lp} + \dfrac{1}{2} \left( b^P \right)^2 \lpa^2 \xi_{\lp \lp} \\
\quad -  n \lp \left(\psin-\psi^{(n-1)}\right)-m \lp \left(\psi^{(m)}-\psi^{(m-1)}\right)\\
\quad = -\alpha \sqrt{\left(1-\rho^2 \right) \left(b^P\right)^2 \lpa^2 \left(\psin_{\lp}\right)^2+ n \lp \left(\psin-\psi^{(n-1)}\right)^2}\\
\quad \quad  -\alpha \sqrt{\left(1-\rho^2 \right) \left(b^P\right)^2 \lpa^2 \left(\psi^{(m)}_{\lp}\right)^2+ m \lp \left(\psi^{(m)}-\psi^{(m-1)}\right)^2}, \\
\xi\left(\li, \lp, T\right)=m+n.
\end{cases}
\end{equation}
Define a differential operator $\L$ on $\G$ by \eqref{operatorL} with $n$ replaced by $m+n$.  It is clear that $\L \eta =0$ on $G$.  Apply the operator $\L$ to $\xi$ to get
\begin{equation} \label{ineq_subadd}
\begin{split}
\L \xi &= n \lp \left(\psin-\psi^{(n-1)}\right) +m \lp \left(\psi^{(m)}-\psi^{(m-1)}\right)\\
&\quad - (m+n) \lp \left(\xi - \psi^{(m+n-1)}\right)\\
&\quad -\alpha \sqrt{\left(1-\rho^2 \right) \left(b^P\right)^2 \lpa^2 \left(\psin_{\lp}\right)^2+ n \lp \left(\psin-\psi^{(n-1)}\right)^2}\\
&\quad  -\alpha \sqrt{\left(1-\rho^2 \right) \left(b^P\right)^2 \lpa^2 \left(\psi^{(m)}_{\lp}\right)^2+ m \lp \left(\psi^{(m)}-\psi^{(m-1)}\right)^2}\\
&\quad + \alpha \sqrt{\left(1-\rho^2 \right) \left(b^P\right)^2 \lpa^2 \xi_{\lp}^2+ (m+ n) \lp \left(\xi - \psi^{(m+n-1)}\right)^2}\\
&\le \left(\psi^{(m+n-1)}-\psi^{(m-1)}-\psi^{(n)}\right)\left(m\lp -\alpha \sqrt{m\lp}\right)\\
&\quad +\left(\psi^{(m+n-1)}-\psi^{(m)}-\psi^{(n-1)}\right)\left(n\lp -\alpha \sqrt{n\lp}\right)\\
&\le 0.
\end{split}
\end{equation}
To get the first inequality in \eqref{ineq_subadd}, we apply Lemma \ref{lem_ABCmn} after assigning $A=\sqrt{\lp}\left(\xi-\psi^{(m+n-1)}\right)$, $B=\sqrt{\lp}\left(\psi^{(m)}-\psi^{(m-1)}\right)$, $C=\sqrt{\lp}\left(\psi^{(n)}-\psi^{(n-1)}\right)$, $B_{\l}=\sqrt{1-\rho^2} \, b^P \cdot \left(\lp-\ulp\right)\psi_{\lp}^{(m)}$, and $C_{\l}=\sqrt{1-\rho^2} \, b^P \cdot \left(\lp-\ulp\right)\psi_{\lp}^{(n)}$. The second inequality in \eqref{ineq_subadd} follows from the induction assumption $\psi^{(m+n-1)} \le \psi^{(k)}+\psi^{(l)}$ with $k+l=m+n-1$, and  from the assumption that $\sqrt{\lp}\ge\alpha$. Since $\xi\left(\li, \lp, T\right)=\eta\left(\li, \lp, T\right)=m+n$, Theorem \ref{thm_comparison} implies that $\eta \le \xi$ on $G$. 
\end{proof}

Property \ref{ppt_subadd} states that our pricing mechanism satisfies subadditivity, $P^{(m+n)} \le P^{(m)}+ \Pn$. This is reasonable  since if subadditivity did not hold, then buyers of pure endowments could purchase separately and thereby save money.

\begin{ppt} \label{ppt_rho0}
Let $\psi^{(n),\rho0}$ be the solution to \eqref{eqn_psin} with $\rho=0$ for $n \ge 0$; then, $\psi^{(n),\rho0} = \psi^{(n),\rho0} (\lp, t)$ is independent of $\li$ and solves the following PDE:
\begin{equation}\label{eqn_Pnrho0}
\begin{cases}
\psi^{(n),\rho0}_t +a^{P,Q} \cdot \lpa  \psi^{(n),\rho0}_{\lp} + \dfrac{1}{2} \left( b^P \right)^2 \lpa^2 \psi^{(n),\rho0}_{\lp \lp}  - n\lp\left(\psi^{(n),\rho0}-\psi^{(n-1),\rho0}\right)\\
\quad = -\alpha \sqrt{ \left(b^P\right)^2 \lpa^2 \left(\psi^{(n),\rho0}_{\lp}\right)^2+n\lp\left (\psi^{(n),\rho0}-\psi^{(n-1),\rho0}\right)^2},\\
\psi^{(n),\rho0} \left( \lp, T \right) = n,
\end{cases}
\end{equation}
with $\psi^{(0),\rho0}\equiv0$ for $n = 0$. 
\end{ppt}

\begin{proof}
The solution of \eqref{eqn_Pnrho0} is independent of $\li$ and also solves \eqref{eqn_psin} when $\rho = 0$.  Uniqueness of the solutions of \eqref{eqn_Pnrho0} and \eqref{eqn_Pn} implies that the solutions of the two PDEs are equal. 
\end{proof}

When $\rho=0$, the optimal investment in the mortality derivative is zero, as we discussed in Remark \ref{rmk_pil0}. Also,  equation \eqref{eqn_Pnrho0} is identical to equation (4.1) of \cite{MilevskyPromislowYoung2005Financial}, which determines the price of $n$ pure endowments in a market without mortality derivatives.  The coincidence of the two results in the case of $\rho=0$ shows that the pricing mechanism we apply is consistent.

It is natural to ask if the hedging will reduce the price of pure endowments. To answer this question, we first make an assumption on $q^{\li}$ to simplify the equation for $\psin$ as follows. 

\begin{ppt}\label{ppt_psin_simp}
When the market price of risk for mortality $q^{\li}$ is independent of $\li$, then $\psin = \psin(\lp, t)$ is also independent of $\li$ and solves the following PDE:
\begin{equation}\label{eqn_psin_simp}
\begin{cases}
\psin_t + a^{P,Q} \cdot \lpa  \psin_{\lp} + \dfrac{1}{2} \left( b^P \right)^2 \lpa^2 \psin_{\lp \lp}  - n\lp\left(\psin - \psi^{(n-1)}\right)\\
\quad = -\alpha \sqrt{\left(1-\rho^2 \right) \left(b^P\right)^2 \lpa^2 \left(\psin_{\lp}\right)^2+n\lp\left (\psin - \psi^{(n-1)}\right)^2},\\
\psin \left( \lp, T\right) = n.
\end{cases}
\end{equation}
\end{ppt}
\begin{proof}
The solution of \eqref{eqn_psin_simp} is independent of $\li$ and also solves \eqref{eqn_psin} when $q^{\li}$ is independent of $\li$.  Uniqueness of the solutions of \eqref{eqn_psin} and \eqref{eqn_psin_simp} implies that the solutions of the two PDEs are equal. 
\end{proof}

Because $\Pn = F \, \psin$,  Property \ref{ppt_psin_simp} implies that if $q^{\li}$ is independent of $\li$, then $\Pn$ is also independent of $\li$.  It follows from this property and Property \ref{ppt_drift} that if the $\qli$ is independent of $\li$, then $\Pn$ increases with increasing market price of mortality risk $\qli$, as one expects.

\begin{ppt}\label{ppt_qli}
Suppose $q^{\li}$ is independent of $\li$ and $\qli_1 \le \qli_2$. Let $\psi^{(n),\qli_i}$ be the solution of \eqref{eqn_psin_simp} with $\qli=\qli_i$, for $i=1,2$ and  for $n\ge0$. Then, $\psi^{(n), \qli_1} \le \psi^{(n),\qli_2}$ on $G$. \end{ppt}
\begin{proof}
From \eqref{muq} we have that $a^{P,Q}_1 \ge a^{P,Q}_2$, and we conclude that  $\psi^{(n), \qli_1} \le \psi^{(n),\qli_2}$ on $G$ from Property \ref{ppt_drift} and Property \ref{ppt_psin_simp}. 
\end{proof}

Next, we give a condition under which hedging with mortality derivatives reduces the price of pure endowments. 

\begin{thm}\label{thm:Pnn}
Suppose $q^{\li}$ is independent of $\li$. Let $\psi^{(n),-}$ denote the solution of \eqref{eqn_psin_simp} with $\rho \, q^{\li} \le 0$, and let $\psi^{(n),0}$ denote the solution of \eqref{eqn_psin_simp} with $\rho=0$. Then, $\psi^{(n),-} \le \psi^{(n),0}$ on $G$.
\end{thm}
\begin{proof}
Define a differential operator $\L$ on $\G$ by \eqref{operatorL} with $g_n$ replaced by
\begin{equation} \label{eq:hatg}
\begin{split}
\hat g_n \left( \lp, t, v, p_2 \right) &= a^{P} \cdot \left(\lp - \ulp\right) p_2 -n \lp \left(v-\psi^{(n-1), 0}\right)  \\
& \quad+\alpha \sqrt{ \left(b^P\right)^2 \lpa^2 p_2^2+ n \lp \left(v-\psi^{(n-1), 0}\right)^2}.
\end{split}
\end{equation}
It is straightforward to check that the function $\hat g_n$ in \eqref{eq:hatg} satisfies the one-sided Lipschitz condition \eqref{Lipschitz} and the growth condtion \eqref{growth}. We have that $\L \psi^{(n),0} = 0$ since $\psi^{(n),0}$ solves \eqref{eqn_psin_simp}  with $\rho=0$.  Apply the operator $\L$ to $\psi^{(n),-}$ to get
\begin{equation}
\begin{split}
\L \psi^{(n),-} &= \rho \, \qli b^{P} \cdot \left(\lp-\ulp\right)\psi^{(n),-}_{\lp} \\
& \quad +\alpha \sqrt{ \left(b^P\right)^2 \lpa^2 \left(\psi^{(n),-}_{\lp}\right)^2+ n \lp \left(\psi^{(n),-}-\psi^{(n-1), 0}\right)^2}\\
& \quad-\alpha \sqrt{\left(1-\rho^2\right)\left(b^P\right)^2 \lpa^2 \left(\psi^{(n),-}_{\lp}\right)^2+ n \lp \left(\psi^{(n),-}-\psi^{(n-1), -}\right)^2} \\
&\ge \alpha \sqrt{n \lp} \left( \psi^{(n-1),0} - \psi^{(n-1),-} \right) \ge 0 = \L\psi^{(n),0}. 
\end{split}
\end{equation}
The first inequality above follows from $\psi^{(n),-}_{\lp} \le 0$, $\rho \, \qli \le 0$, and Lemma \ref{lem:ABC}. The second inequality follows by an induction step; recall that $\psi^{(0),0} = \psi^{(0), -} = 0$. Additionally, $\psi^{(n), -}\left(\lp,T\right)=\psi^{(n), 0}\left(\lp,T\right)$, so Theorem \ref{thm_comparison} implies that $\psi^{(n)-} \le \psi^{(n),0}$ on $G$.
\end{proof}

\begin{remark}
One can interpret the price $\Pn$ with $\rho = 0$ as the price for which no hedging with the mortality derivative is allowed because the optimal investment in the mortality derivative when $\rho = 0$ is 0, which follows from Property \ref{ppt_rho0}. Thus, Theorem \ref{thm:Pnn} asserts that when $\rho \, \qli \le 0$, the price when hedging is allowed is less than the price with no hedging.  However, if $\rho \, \qli > 0$, then we cannot conclude that hedging necessarily reduces the price of the pure endowment.  We discuss this more fully at the end of the next section.
\end{remark}

%-----------------------------------------------------------------------------------------------------------------------------------------------------------------------------------------------------------------------------------------------------

%\subsection{\label{sec:gooddeal}Relation with Good Deal Bound}
%-----------------------------------------------------------------------------------------------------------------------------------------------------------------------------------------------------------------------------------------------------

\section{\label{sec:limit} Limiting Behavior of $\dfrac{1}{n}P^{(n)}$ as $n\to \infty$}

In this section, we consider the limiting behavior of $\frac{1}{n} P^{(n)}$. First, we show that the price per risk, $\frac{1}{n}P^{(n)}$, decrease as $n$ increases;  that is, by increasing the number of pure endowment contracts, we reduce the price per contract. Then, we further explore how far $\frac{1}{n}P^{(n)}$ decreases by determining the limiting value of the decreasing sequence $\left\{ \frac{1}{n}P^{(n)} \right\}$.  Surprisingly, we find in Theorem \ref{thm_limit} that the limiting value solves a {\it linear} PDE.  The proofs of most results in this section are modifications of the proofs given by \cite{MilevskyPromislowYoung2005Financial}.

To prove the limiting properties of  $\frac{1}{n} P^{(n)}$, we use the Lemma $4.12$ in \cite{MilevskyPromislowYoung2005Financial}.  We restate this lemma without proof. 

\begin{lemm} \label{lem_ACBlp}
If $n \ge 2$, and if $A \ge C \ge 0$ and $B_{\l}$ are constants, then the following inequality holds 
\begin{equation} 
\sqrt{B_{\l}^2+\dfrac{1}{n}C^2} \; \le \; \sqrt{n-2} \, (A-C) +\sqrt{B_{\l}^2+\dfrac{1}{n-1}\left[(n-1)C-(n-2)A\right]^2} \, .
\end{equation}
\end{lemm}

\begin{prop}\label{lem_Pn_dec}
$\frac{1}{n}\Pn$ decreases with respect to $n$ for $n \ge 1$. 
\end{prop}

\begin{proof}
It is sufficient to show that $\frac{1}{n} \psin$ decreases with respect to $n$.  Define $\phi^{(n)} \triangleq \frac{1}{n}\psin$, and we will show that $\phi^{(n-1)} \ge \phi^{(n)}$ for $n\ge2$ by induction.  From \eqref{eqn_psin}, we deduce that $\phi^{(n)}$  solves
\begin{equation}\label{eqn_phin} 
\begin{cases}
\phin_t+a^{I,Q} \cdot \lia \phin_{\li} + a^{P,Q} \cdot \lpa  \phin_{\lp} +\dfrac{1}{2}\left(b^I \right)^2 \lia^2 \phin_{\li \li} \\
\quad + \rho \, b^I \, b^P \lia \lpa \phin_{\li \lp} + \dfrac{1}{2} \left( b^P \right)^2 \lpa^2 \phin_{\lp \lp}  \\
\quad -  \lp \left[n\phin-(n-1)\phi^{(n-1)}\right]\\
\quad = -\alpha \sqrt{\left(1-\rho^2 \right) \left(b^P\right)^2 \lpa^2 \left(\phin_{\lp}\right)^2+ \dfrac{1}{n} \lp\left[n\phin-(n-1)\phi^{(n-1)}\right]^2},\\
\phin \left(\li, \lp, T\right) = 1, 
\end{cases}
\end{equation}
with $\phi^{(1)} = \psi$, in which $\psi$ solves \eqref{eqn_psi}.

We first show that $\phi^{(1)} \ge \phi^{(2)} $. To this end, we define a differential operator $\L$ on $\G$ by \eqref{operatorL} with $g_n$ replaced by 
\begin{equation}\label{g_phi_2}
\begin{split}
\hat g_2(\li,\lp,t,v, p_1,p_2) &= a^{I,Q} \cdot \lia p_1 + a^{P,Q} \cdot \lpa  p_2 -  \lp \left(2v-\psi\right)\\
&\quad  + \alpha \sqrt{\left(1-\rho^2 \right) \left(b^P\right)^2 \lpa^2 \left(p_2\right)^2+ \dfrac{1}{2} \lp\left( 2v-\psi\right)^2}.
\end{split}
\end{equation}
It is clear that $\hat g_2$ satisfies conditions \eqref{Lipschitz} and \eqref{growth}; hence, we can apply Theorem \ref{thm_comparison}.  Note that $\L \phi^{(2)}=0$ since $\phi^{(2)}$ solves \eqref{eqn_phin} with $n=2$. By applying the operator $\L$ to $\phi^{(1)} = \psi$, we get
\begin{equation}
\begin{split}
\L \phi^{(1)}&=\alpha \sqrt{\left(1-\rho^2 \right) \left(b^P\right)^2 \lpa^2 \psi_{\lp}^2+ \dfrac{1}{2} \lp \psi^2}\\
&\quad -\alpha \sqrt{\left(1-\rho^2 \right) \left(b^P\right)^2 \lpa^2 \psi_{\lp}^2+ \lp\psi^2}\\
&\le 0 = \L \phi^{(2)}.
\end{split}
\end{equation}
Since $\phi^{(1)}\left(\li, \lp, T\right)=\phi^{(2)}\left(\li, \lp, T\right)=1$, Theorem \ref{thm_comparison} implies that $\phi^{(1)}\ge \phi^{(2)}$ on $G$. 

Assume that for $n\ge3$, $\phi^{(n-2)} \ge \phi^{(n-1)}$ on $G$, and we show that $\phi^{(n-1)} \ge \phin$.  We define a differential operator $\L$ on $\G$ by \eqref{operatorL} with $g_n$ replaced by 
\begin{equation}\label{g_phi_n}
\begin{split}
\hat g_n(\li,\lp,t,v, p_1,p_2) = a^{I,Q} \cdot \lia p_1 + a^{P,Q} \cdot \lpa  p_2 - \lp \left[nv-(n-1)\phi^{(n-1)}\right]\\
\quad  + \alpha \sqrt{\left(1-\rho^2 \right) \left(b^P\right)^2 \lpa^2 p_2^2+ \dfrac{1}{n} \lp\left[nv-(n-1)\phi^{(n-1)}\right]^2}.
\end{split}
\end{equation}
It is clear that $\hat g_n$ satisfies conditions \eqref{Lipschitz} and \eqref{growth}; hence, we can apply Theorem \ref{thm_comparison}.  Note that $\L \phi^{(n)}=0$ since $\phi^{(n)}$ solves \eqref{eqn_phin}. Apply the operator $\L$ to $\phi^{(n-1)}$ to get 
\begin{equation}\label{temp1}
\begin{split}
\L \phi^{(n-1)} & = (n-2) \lp \left( \phi^{(n-1)} - \phi^{(n-2) }\right)+\alpha \sqrt{\left(1-\rho^2 \right) \left(b^P\right)^2 \lpa^2 \left(\phi^{(n-1)}_{\lp}\right)^2+ \dfrac{1}{n} \lp\left(\phi^{(n-1)}\right)^2}\\
&\quad -\alpha \sqrt{\left(1-\rho^2 \right) \left(b^P\right)^2 \lpa^2 \left(\phi^{(n-1)}_{\lp}\right)^2+ \dfrac{1}{n-1} \lp\left[(n-1)\phi^{(n-1)}-(n-2)\phi^{(n-2)}\right]^2}\\
& \le \left[(n-2)\lp -\alpha \sqrt{(n-2)\lp} \right]\left(\phi^{(n-1)}-\phi^{(n-2)}\right) \le 0 = \L \phin. 
\end{split}
\end{equation}
To get the first inequality in \eqref{temp1}, we use Lemma \ref{lem_ACBlp} by assigning  $A=\sqrt{\lp} \, \phi^{(n-2)}$,  $C=\sqrt{\lp} \, \phi^{(n-1)}$, and $B_{\l}=\sqrt{1-\rho^2} \, b^P \cdot \left(\lp-\ulp\right)\phi_{\lp}^{(n-1)}$. We also use the induction assumption that  $\phi^{(n-2)} \ge \phi^{(n-1)}$.  Additionally,  $\phi^{(n-1)}\left(\li, \lp, T\right)=\phi^{(n)}\left(\li, \lp, T\right)=1$, and Theorem \ref{thm_comparison} implies that $\phi^{(n-1)} \ge \phin$ on $G$. 
\end{proof} 

In what follows, we answer the question inspired by Proposition \ref{lem_Pn_dec}, namely, what is the {\it limit} of the non-negative, decreasing sequence $\left\{ \frac{1}{n} \Pn \right\}$?  In Theorem \ref{thm_limit} below, we will show the limit equals $F \, \b$, in which $\beta = \beta\left(\li, \lp, t \right)$ denote the solution of the following PDE:
\begin{equation}\label{eqn_beta}
\begin{cases}
\b_t + a^{I,Q} \cdot \lia \b_{\li} + \left[a^{P,Q} - \alpha \sqrt{1-\rho^2} \, b^P\right] \lpa  \b_{\lp}  + \dfrac{1}{2}\left(b^I \right)^2 \lia^2 \b_{\li \li} \\
\quad   + \rho \, b^I \, b^P \lia \lpa \b_{\li \lp} + \dfrac{1}{2} \left( b^P \right)^2 \lpa^2 \b_{\lp \lp}  -  \lp \b \\
\quad  = 0, \\
\b\left(\li, \lp, T\right)=1. 
\end{cases}
\end{equation}
By applying the Feyman-Kac Theorem to \eqref{eqn_beta}, we obtain an expression for $\beta$ as an expectation:
\begin{equation}
\beta(\li, \lp, t) =\E^{\tilde{Q}} \left[e^{-\int_t^T \lp_s \d s} \, \bigg| \, \li_t = \li, \lp_t=\lp \right], 
\end{equation}
in which the $\tilde{Q}$-dynamics of $\left\{ \li_t \right\}$ and $\left\{ \lp_t \right\}$ follow, respectively,
\begin{equation} \label{eqn:Qtilde_li}
\d \li_t =a^{I, Q}(\li_t, t)\left(\li_t-\uli \right)\d t+b^I(t) \left(\li_t - \uli \right)\d W^{I, Q}_t
\end{equation}
and
\begin{equation} \label{eqn:Qtilde_lp}
\d \lp_t = \left[ a^{P,Q}(\li_t, \l^P_t, t) - \alpha \sqrt{1-\rho^2} \, b^P(t) \right]  \left(\l^P_t-\ulp \right) \d t + b^P(t) \left(\lp_t - \ulp \right) \d \tilde W^{P, Q}_t.
\end{equation} 
Here, $\tilde W^{P, Q}_t = W^{P, Q}_t + \alpha \sqrt{1 - \rho^2} \, t$.

We begin by proving that $\frac{1}{n}\Pn$ is bounded below by $F \, \beta$, and for that purpose, we need the following lemma. 

\begin{lemm}
The function $\b$ defined by \eqref{eqn_beta} is non-increasing with respect to $\lp$. 
\end{lemm}

\begin{proof}
Denote $f=\b_{\lp}$, and we deduce from  \eqref{eqn_beta}  that $f$ solves the following PDE:
\begin{equation}\label{eqn_beta_lp}
\begin{cases}
f_t + \left[ a^{I,Q} + \rho \, b^I \, b^P \right] \lia f_{\li} + \left[ a^{P,Q} - \alpha \sqrt{1-\rho^2} \, b^P \right]  \lpa f_{\lp} \\
\quad + \left[a^{P,Q}_{\lp} \cdot \lpa + a^{P,Q} - \alpha \sqrt{1-\rho^2} \, b^P -\lp \right] f+\dfrac{1}{2}\left(b^I \right)^2 \lia^2 f_{\li \li}\\
\quad  +\rho \, b^I \, b^P \lia \lpa f_{\li\lp} + \dfrac{1}{2} \left( b^P \right)^2 \lpa^2 f_{\lp \lp}  -  \b\\
\quad  = 0, \\
f\left(\li, \lp, T\right) = 0.
\end{cases}
\end{equation}
Define a differential operator $\L$ on $\G$ by \eqref{operatorL} with $g_n$ replaced by
\begin{equation}\label{eqn_g_beta}
\begin{split}
\tilde g(\li,\lp,t,v,p_1,p_2)& = \left[ a^{I,Q} + \rho \, b^I \, b^P \right] \lia p_1 + \left[ a^{P,Q} - \alpha \sqrt{1-\rho^2} \, b^P\right]  \lpa p_2 \\
& \quad + \left[a^{P,Q}_{\lp} \cdot \lpa + a^{P,Q} - \alpha \sqrt{1-\rho^2} \, b^P -\lp \right] v - \b\\
\end{split}
\end{equation}
Because of  Assumption \ref{assum_growth}, it is straightforward to check that the function $\tilde g$ in \eqref{eqn_g_beta} satisfies the one-sided Lipschitz condition \eqref{Lipschitz} and the growth condtion \eqref{growth}.  Because $f$ solves \eqref{eqn_beta_lp}, we have that $\L f=0$.  Because $\beta$ is clearly non-negative, $\L \0 = -\beta \le 0$, in which $\0$ is the constant function of $0$ on $G$.  Additionally, $f\left(\li, \lp, T\right)=0$, so Theorem \ref{thm_comparison} implies that $\beta_{\lp}\le 0$. 
\end{proof}

\begin{lemm}\label{lem_beta}
For $n\ge 1$, $\frac{1}{n}\Pn \ge F \beta$, in which $\beta$ is given in \eqref{eqn_beta}
\end{lemm}

\begin{proof}
It is sufficient to show that $\frac{1}{n}\psin \ge \beta$ on $G$. We prove this property by induction. First, for $n=1$, we show that $\b \le \psi^{(1)} = \psi$. Define a differential operator $\L$ on $\G$ by \eqref{operatorL} with $n=1$.  Recall that $\psi^{(0)}=0$ in \eqref{eqn_gn}.  Since $\psi$ solves \eqref{eqn_psi}, $\L\psi=0$. Also, 
\begin{equation}
\L \beta = \alpha \sqrt{\left(1-\rho^2 \right) \left(b^P\right)^2 \lpa^2 \b_{\lp}^2+  \lp \beta ^2} - \alpha \sqrt{1-\rho^2} b^P \lpa \left| \b_{\lp} \right| \ge 0 = \L\psi.
\end{equation}
Additionally, $\b\left(\li, \lp, T\right) = 1 = \psi\left(\li, \lp, T\right)$, so Theorem \ref{thm_comparison} implies that $\b \le \psi \ge$ on $G$.  

For $n\ge 1$,  assume that $\b \le \phi^{(n-1)}$ and show that $\b \le \phi^{(n)}$, in which $\phin = \frac{1}{n}\psi^{(n)}$ for $n\ge1$, as we defined in the proof of Proposition \ref{lem_Pn_dec}.  Define a differential operator  $\L$ by \eqref{operatorL} with $g_n$ replaced by $\hat g_n$ given by \eqref{g_phi_n}.  Since $\phi^{(n)}$ solves \eqref{eqn_phin}, $\L \phin=0$. By applying this operator on $\b$, we get 
\begin{equation}
\begin{split}
\L \b & =  \alpha \sqrt{\left(1-\rho^2 \right) \left(b^P\right)^2 \lpa^2 \left(\b_{\lp}\right)^2 + \dfrac{1}{n} \lp\left[n\b-(n-1)\phi^{(n-1)}\right]^2}\\
&\quad -\alpha \sqrt{1-\rho^2} \, b^P \lpa \left|\b_{\lp}\right| + \lp \left[(n-1)\phi^{(n-1)}-(n-1)\b\right]\\
& \ge 0 =\L \phin.
\end{split}
\end{equation}
Also, $\beta\left(\li, \lp,T\right)=\phin\left(\li, \lp,T\right)=1$; thus, Theorem \ref{thm_comparison} implies that $\b \le \phin =\frac{1}{n}\psin$ on $G$. 
\end{proof}

Next, we show that $\lim_{n\to \infty}\frac{1}{n}\Pn = F \beta$. To this end, we need some auxiliary results. First, we prove that $\psin$ is bounded from above by $\gamma^{(n)} = \gamman (\li, \lp, t)$  for $n\ge0$, in which the function $\gamma^{(n)}$ solves the following PDE:
\begin{equation}\label{eqn_gamman}
\begin{cases}
\gamman_t + a^{I,Q} \cdot \lia\gamman_{\li} + \left[ a^{P,Q} - \alpha \sqrt{1-\rho^2} \, b^P \right] \lpa  \gamman_{\lp}  + \dfrac{1}{2}\left(b^I \right)^2 \lia^2 \gamman_{\li \li} \\
\quad + \rho \, b^I \, b^P \lia \lpa \gamman_{\li \lp} + \dfrac{1}{2} \left( b^P \right)^2 \lpa^2 \gamman_{\lp \lp} \\
\quad - \left(n\lp -\alpha \sqrt{n\lp}\right) \left(\gamman-\gamma^{(n-1)}\right) = 0, \\
\gamman\left(\li, \lp, T\right)=n,
\end{cases}
\end{equation}
in which $\gamma^{(0)} \equiv 0$. 

\begin{lemm}\label{lem_gamman}
The function $\gamman$ given by \eqref{eqn_gamman} is non-increasing with respect to $\lp$, and $\gamma^{(n)}\ge\gamma^{(n-1)}$ for $n\ge 1$ on $G$.  
\end{lemm}
\begin{proof}
The proof that $\gamman_{\lp}\le 0$ is similar to the proof that $\psin_{\lp}\le 0$ in Property \ref{ppt_wrtlp}. Also, the proof that $\gamma^{(n)}\ge\gamma^{(n-1)}$ is similar to the proof that $\psi^{(n)}\ge\psi^{(n-1)}$ in Property \ref{ppt_mono}. Therefore, we omit the details of the proof. 
\end{proof}

\begin{lemm}\label{lem_gamma}
For $n\ge0$, $\gamman\ge \psin$ on $G$. 
\end{lemm}
\begin{proof}
We prove this lemma by induction.  For $n=0$, we have $\gamma^{(0)}=\psi^{(0)}=0$. Assume that for $n\ge1$, we have $\gamma^{(n-1)}\ge \psi^{(n-1)}$, and show that $\gamman\ge \psin$.  For this purpose, define a differential operator $\L$ on $\G$ by \eqref{operatorL}.  Then,  $\L\psi=0$, and 
\begin{equation}
\begin{split}
\L \gamman &= \alpha \sqrt{1-\rho^2} \, b^P\lpa  \gamman_{\lp} + \left(n\lp -\alpha \sqrt{n\lp}\right) \left(\gamman-\gamma^{(n-1)}\right) -n \lp \left(\gamman-\psi^{(n-1)} \right) \\
& \quad+\alpha \sqrt{\left(1-\rho^2 \right) \left(b^P\right)^2 \lpa^2\left( \gamman_{\lp}\right)^2+ n \lp \left(\gamman-\psi^{(n-1)}\right)^2}\\
&\le - \left(n\lp -\alpha \sqrt{n\lp}\right) \left(\gamma^{(n-1)}-\psi^{(n-1)}\right) \le 0 =\L \psin.
\end{split}
\end{equation}
The first inequality above is due to the fact that $\gamman_{\lp}\le 0$,  that $\gamman\ge\gamma^{(n-1)} \ge \psi^{(n-1)}$, and that $\sqrt{A^2+B^2}\le |A|+|B|$. Additionally, we have that $\gamman\left(\li, \lp, T\right)=\psin\left(\li, \lp, T\right)=n$; then, Theorem \ref{thm_comparison} implies that $\gamman\ge \psin$ on $G$. 
\end{proof}
 
 Next, we prove the main result of this section. 
 
\begin{thm}\label{thm_limit}
$\lim_{n\to \infty}\dfrac{1}{n}\Pn(r, \li,\lp,t)=F(r, t) \, \beta\left(\li, \lp, t \right)$ on $G$. 
\end{thm}
\begin{proof}
By Lemmas \ref{lem_beta} and \ref{lem_gamma}, it is sufficient to show that $\lim_{n \to \infty} \left(\frac{1}{n}\gamman -\beta \right)= 0$ since $\frac{1}{n}\gamman - \beta \ge \frac{1}{n}\psin -\beta \ge 0$. 
For $n \ge 1$, define $\Gamman$ on $G$ by $\Gamma^{(n)} = \frac{1}{n}\gamman -\beta$, so we just need to prove that $\lim_{n \to \infty} \Gamman= 0$.  For $n\ge1$, the function $\Gamman$ solves the following PDE:
\begin{equation}\label{eqn_Gamman}
\begin{cases}
\Gamman_t + a^{I,Q} \cdot \lia\Gamman_{\li} + \left[a^{P,Q} - \alpha \sqrt{1-\rho^2} \, b^P \right] \lpa  \Gamman_{\lp}  \\
\quad   + \dfrac{1}{2}\left(b^I \right)^2 \lia^2 \Gamman_{\li \li} + \rho \, b^I \, b^P \lia \lpa \Gamman_{\li \lp} + \dfrac{1}{2} \left( b^P \right)^2 \lpa^2 \Gamman_{\lp \lp} \\
\quad   - \left(n\lp -\alpha \sqrt{n\lp}\right) \Gamman\\
\quad  = -\alpha\sqrt{\dfrac{\lp}{n}}\b-(n-1)\left(\lp -\alpha \sqrt{\dfrac{\lp}{n}}\right)\Gamma^{(n-1)}, \\
\Gamman\left(\li, \lp, T\right)=0,
\end{cases}
\end{equation}
with $0\le \Gamma^{(1)}=\gamma^{(1)}-\b \le 1$ on $G$.  By applying the Feyman-Kac Theorem to \eqref{eqn_Gamman}, we obtain the following expression for $\Gamman$ in terms of $\Gamma^{(n-1)}$:
\begin{equation}
\begin{split}
\Gamman&(\li,\lp,t) =\alpha \, \E^{\tilde{Q}}\left[\int_t^T\sqrt{\dfrac{\lp_s}{n}} \, \b\left(\li_s,\lp_s,s\right) e^{-\int_t^s\left(n\lp_u-\alpha\sqrt{n\lp_u}\right)\d u} \,  \d s \, \bigg| \, \li_t = \li, \lp_t = \lp \right]\\
& + (n-1) \, \E^{\tilde{Q}}\left[\int_t^T\left(\lp_s-\alpha\sqrt{\dfrac{\lp_s}{n}}\right)\Gamma^{(n-1)} e^{-\int_t^s\left(n\lp_u-\alpha\sqrt{n\lp_u}\right)\d u} \, \d s \, \bigg| \, \li_t = \li,  \lp_t = \lp \right], \\
\end{split}
\end{equation}
in which the $\tilde{Q}$-dynamics of $\left\{ \li_t \right\}$ and $\left\{ \lp_t \right\}$ follow, respectively, equations \eqref{eqn:Qtilde_li} and \eqref{eqn:Qtilde_lp}.

Suppose $\Gamma^{(n-1)} \le K_{n-1}$ on $G$ for some $n\ge2$ and for some constant $K_{n-2} \ge 0$.  Note that $\b\le 1$ on $G$, so we get the following inequality:
\begin{equation}\label{temp_ineq}
\begin{split}
\Gamman&(\li,\lp,t) \le \alpha \, \E^{\tilde{Q}}\left[\int_t^T\sqrt{\dfrac{\lp_s}{n}} \, e^{- \int_t^s \left(n\lp_u-\alpha\sqrt{n\lp_u}\right)\d u} \, \d s \, \bigg| \, \li_t = \li, \lp_t = \lp \right]\\
& + (n-1) K_{n-1} \, \E^{\tilde{Q}}\left[\int_t^T\left(\lp_s-\alpha\sqrt{\dfrac{\lp_s}{n}}\right) e^{- \int_t^s \left(n\lp_u-\alpha\sqrt{n\lp_u}\right)\d u} \, \d s \, \bigg| \, \li_t = \li, \lp_t = \lp \right]. \\
\end{split}
\end{equation}
Equivalently, we can write the inequality \eqref{temp_ineq} as
\begin{equation}
\Gamman(\li,\lp,t)\le \dfrac{1}{n^{3/2}}A^{(n)}\left(\li,\lp,t \right) + \dfrac{n-1}{n}K_{n-1}B^{(n)}\left(\li,\lp,t \right),
\end{equation}
in which the functions $A^{(n)}$ and $B^{(n)}$ are defined as
\begin{equation}\label{eqn_An}
A^{(n)}\left(\li,\lp,t \right) = \alpha \, \E^{\tilde{Q}}\left[\int_t^T n \sqrt{\lp_s} \, e^{- \int_t^s\left(n\lp_u-\alpha\sqrt{n\lp_u}\right) \d u} \, \d s \, \bigg| \, \li_t = \li, \lp_t=\lp \right],
\end{equation}
and 
\begin{equation}\label{eqn_Bn}
B^{(n)}\left(\li,\lp,t \right) = \E^{\tilde{Q}}\left[\int_t^T\left(n\lp_s-\alpha\sqrt{n\lp_s}\right) e^{- \int_t^s\left(n\lp_u-\alpha\sqrt{n\lp_u}\right)\d u} \, \d s \, \bigg| \, \li_t = \li, \lp_t=\lp \right].
\end{equation}
After the next two lemmas that give us bounds on $A^{(n)}$ and $B^{(n)}$, respectively, we finish the proof of Theorem \ref{thm_limit}.
\end{proof}

\begin{lemm}\label{lem_An}
For $n\ge2$, $A^{(n)}\le J = \dfrac{\alpha \sqrt{2}}{\sqrt{2 \ulp}-\alpha}$ on $G$, in which $A^{(n)}$ is defined in \eqref{eqn_An}.
\end{lemm}

\begin{proof}
By the Feyman-Kac Theorem, $A^{(n)}$ in \eqref{eqn_An} solves the following PDE
\begin{equation}
\begin{cases}
A^{(n)}_t + a^{I,Q} \cdot \lia A^{(n)}_{\li} + \left[ a^{P,Q} - \alpha \sqrt{1-\rho^2} \, b^P \right] \lpa A^{(n)}_{\lp} + \dfrac{1}{2}\left(b^I \right)^2 \lia^2 A^{(n)}_{\li \li}   \\
\quad   + \rho \, b^I \, b^P \lia \lpa A^{(n)}_{\li \lp} + \dfrac{1}{2} \left( b^P \right)^2 \lpa^2 A^{(n)}_{\lp \lp}  - \left(n\lp -\alpha \sqrt{n\lp}\right) A^{(n)}\\
\quad  = -\alpha n \sqrt{\lp}, \\
A^{(n)}\left(\li, \lp, T\right)=0. 
\end{cases}
\end{equation}
For $n\ge2$, we define a differential operator  $\L$ by \eqref{operatorL} with $g_n$ replaced by 
\begin{equation}
\begin{split}
\tilde g_n(\li,\lp,t,v, p_1,p_2)& = a^{I,Q} \cdot \lia p_1 + \left[ a^{P,Q} - \alpha \sqrt{1-\rho^2} \, b^P \right] \lpa p_2\\
&\quad - \left(n\lp -\alpha \sqrt{n\lp}\right) v +  \alpha n \sqrt{\lp}.\\
\end{split}
\end{equation}
Since $\tilde g_n$ satisfies conditions \eqref{Lipschitz} and \eqref{growth},  we can apply Theorem \ref{thm_comparison}.  It is clear that $\L A^{(n)}=0$, and by applying the operator $\L$ to $\bf J$, the function that is identically equal to $J$, we get
\begin{equation}
\L {\bf J} = - \left(n\lp -\alpha \sqrt{n\lp}\right) J +  \alpha n \sqrt{\lp} \le 0= \L A^{(n)}. 
\end{equation}
Since $A^{(n)}\left(\li, \lp, T\right)=0\le J$, Theorem \ref{thm_comparison} implies that $A^{(n)}\le J$ on $G$. 
\end{proof}

\begin{lemm}\label{lem_Bn}
For $n\ge2$, $B^{(n)} \le 1$ on $G$, in which $B^{(n)}$ is defined in \eqref{eqn_Bn}.
\end{lemm}

\begin{proof}
By the Feyman-Kac Theorem, $B^{(n)}$ in \eqref{eqn_Bn} solves the following PDE
\begin{equation}
\begin{cases}
B^{(n)}_t + a^{I,Q} \cdot \lia B^{(n)}_{\li} + \left[ a^{P,Q} - \alpha \sqrt{1-\rho^2} \, b^P \right] \lpa B^{(n)}_{\lp} + \dfrac{1}{2}\left(b^I \right)^2 \lia^2 B^{(n)}_{\li \li}   \\
\quad   + \rho \, b^I \, b^P \lia \lpa B^{(n)}_{\li \lp} + \dfrac{1}{2} \left( b^P \right)^2 \lpa^2 B^{(n)}_{\lp \lp}  - \left(n\lp -\alpha \sqrt{n\lp}\right) B^{(n)}\\
\quad  = -\left(n \lp-\alpha \sqrt{n\lp} \right),\\
B^{(n)}\left(\li, \lp, T\right) = 0. 
\end{cases}
\end{equation}
For $n\ge2$, we define a differential operator $\L$ on $\G$ by \eqref{operatorL} with $g_n$ replaced by 
\begin{equation}
\begin{split}
\hat g_n(\li,\lp,t,v, p_1,p_2)& = a^{I,Q} \cdot \lia p_1 + \left[ a^{P,Q} - \alpha \sqrt{1-\rho^2} \, b^P \right] \lpa p_2\\
&\quad - \left(n\lp -\alpha \sqrt{n\lp}\right) v+ \left(n\lp-\alpha \sqrt{n\lp} \right).\\
\end{split}
\end{equation}
Since $\hat g_n$ satisfies conditions \eqref{Lipschitz} and \eqref{growth},  we can apply Theorem \ref{thm_comparison}.  It is clear that $\L B^{(n)}=0$, and by applying the operator $\L$ to $\1$, we get
$\L \1= - \left(n\lp -\alpha \sqrt{n\lp}\right) +  \left(n\lp-\alpha \sqrt{n\lp} \right) = 0= \L B^{(n)}$.  Since $B^{(n)}\left(\li, \lp, T\right) = 0 \le 1$, Theorem \ref{thm_comparison} implies that $B^{(n)} \le 1$ on $G$. 
\end{proof}

\noindent {\it End of Proof of Theorem \ref{thm_limit}.} By Lemmas \ref{lem_An} and \ref{lem_Bn}, we get the following result: for $n\ge2$, if $\Gamma^{(n-1)} \le K_{n-1}$, then 
\begin{equation}
\Gamman\le K_n \triangleq \dfrac{J}{n^{3/2}}+\dfrac{n-1}{n}K_{n-1}, 
\end{equation}
with $K_1=1$. 
Define $L_n = nK_n$ and note that $L_n=L_{n-1}+\dfrac{J}{\sqrt{n}}$ for $n\ge2$. It follows that 
\begin{equation}
L_n=1+\sum_{i=2}^n \dfrac{J}{\sqrt{i}}\le 1+ J \int_1^n \dfrac{\d x}{\sqrt{x}}\le 1+ 2J\sqrt{n}, \quad n\ge 2, 
\end{equation}
which implies  that on $G$,
\begin{equation}
\Gamman\le K_n \le \dfrac{1}{n}+ \dfrac{2J}{\sqrt{n}}, \quad n\ge 1. 
\end{equation}
$\lim_{n \to \infty} \frac{1}{n}+ \frac{2J}{\sqrt{n}} = 0$; therefore, $\Gamman$ converges to $0$ uniformly on $G$ as $n$ goes to infinity. In other words, $\lim_{n\to \infty}\frac{1}{n}\Pn = F \, \beta$ on $G$.

We end this section with some properties of $\beta$ with the goal of determining the effect of $\rho$ on $\beta$.

\begin{ppt}\label{ppt_beta_simp}
If $q^{\li}$ is independent of $\li$, then $\b = \b(\lp, t)$ is independent of $\li$ and solves the following PDE:
\begin{equation}\label{eqn_beta_simp}
\begin{cases}
\b_t+\left[a^{P,Q} - \alpha \sqrt{1-\rho^2} \, b^P\right] \lpa  \b_{\lp} + \dfrac{1}{2} \left( b^P \right)^2 \lpa^2 \b_{\lp \lp} -  \lp \b = 0, \\
\b\left(\lp, T\right)=1. 
\end{cases}
\end{equation}
\end{ppt}
\begin{proof}
The solution of \eqref{eqn_beta_simp} is independent of $\li$ and also solves \eqref{eqn_beta} when $q^{\li}$ is independent of $\li$. Uniqueness of the solutions of \eqref{eqn_beta} and \eqref{eqn_beta_simp} implies that solutions of the two PDEs are equal. 
\end{proof}

\begin{thm} \label{thm:limcon}
Suppose $q^{\li}$ is independent of $\li$, and define $\hat a \triangleq a^P - \left[ \rho \, q^{\li} + \alpha \sqrt{1-\rho^2} \right] b^P$. Let $\b^{\hat a_i}$ denote the solution of \eqref{eqn_beta_simp} with $\hat a = \hat a_i$, for $i=1,2$. Then, $\b^{\hat a_1} \ge \b^{\hat a_2}$ on $G$ if $\hat a_1 \le \hat a_2$.
\end{thm}

\begin{proof}
Define a differential operator $\L$ on $\G$  by \eqref{operatorL} with $g_n$ replaced by
\begin{equation}\label{g_beta}
\hat g(\li, v,  p) = \hat{a}_1 \cdot \lpa p- \lp v.
\end{equation}
It is straightforward to check that the function $\hat g$ in \eqref{g_beta} satisfies the one-sided Lipschitz condition \eqref{Lipschitz} and the growth condition \eqref{growth}. Since $\b^{\hat{a}_1}$ solves \eqref{eqn_beta_simp} with $\hat{a}=\hat{a}_1$, we have that $\L \b^{\hat a_1}=0$. Apply this operator on $\b^{\hat{a}_2}$ to obtain
\begin{equation}
\L \b^{\hat a_2} = \left(\hat{a}_1 - \hat{a}_2 \right)\lpa\b^{\hat{a}_2}_{\lp} \ge 0 =\L \b^{\hat{a}_1}. 
\end{equation}
Since $\b^{\hat{a}_1}\left(\lp, T\right) = \b^{\hat{a}_2}\left( \lp, T\right)=1$, Theorem \ref{thm_comparison} implies that $\b^{\hat{a}_1} \ge \b^{\hat{a}_2}$ on $G$.
\end{proof}

\begin{remark}  \label{rem:4_1}
When $\rho=1$, namely the the insured individuals and the reference population  face the same uncertainty in their respective hazard rates, the limiting price per contract is reduced by hedging when $\qli$ is less than the pre-specified instantaneous Sharpe ratio $\alpha$.  Indeed, the drift $\hat a = a^P - \left[ \rho \, q^{\li} + \alpha \sqrt{1-\rho^2} \right] b^P$ from Theorem \ref{thm:limcon} equals $a^P - q^{\li} \, b^P$ when $\rho = 1$.  Also, the effect of not allowing hedging can be achieved by setting $\rho = 0$ throughout our work, as discussed in Remark \ref{rmk_pil0}; in that case, the drift $\hat a$ becomes $a^P - \alpha \, b^P$. Thus, according to Theorem \ref{thm:limcon}, the limiting price per contract is reduced when hedging is allowed if $\qli < \alpha$.

In other words, hedging with mortality derivative benefits the insured, through a reduced price, when the market price of mortality risk is lower than that required by the insurance company.  In this limiting case, the risks inherent in the contract can be fully hedged using the interest rate derivative and the mortality derivative. Indeed, the variance of the hedging portfolio goes to $0$ as $n$ goes to infinity when $\rho=1$. Refer to Remark \ref{rmk:21} in which we discuss the mortality risk in the single-life case. So, as $n\to\infty$, the risk coming from the timing of the deaths disappears; compare with \eqref{eq:227}.

 The price of the contract is reduced by transferring the mortality risk to a counterparty who requires a lower compensation for the risk than the insurance company does.  By contrast, for a single pure endowment contract, the volatility in the contract due to the uncertainty of the individual's time of death is not  hedgeable with mortality derivatives  even when $\rho=1$.  In the single-life case, even if $\qli$ is less than $\alpha$, hedging does not guarantee a reduction of the contract price. 
\end{remark}

\begin{cor}\label{cor:41} 
Suppose $q^{\li}$ is independent of $\li$, and let $\b^{a_i}$ denote the solution of \eqref{eqn_beta_simp} with $a^P=a_i$, for $i=1,2$. Then, $\b^{a_1} \ge \b^{a_2}$ on $G$ if $a_1 \le a_2$.
\end{cor}

The result above is consistent with our intuition.  Indeed, with a higher drift on the hazard rate, the individual is less likely to survive to time $T$, and, consequently, the (limiting) value of the pure endowment contract is lower. 

\begin{cor}\label{cor_comprho}
Suppose $q^{\li}$ is independent of $\li$, and let $\b^{\rho_i}$ denote the solution of \eqref{eqn_beta_simp} with $\rho =\rho_i $ for $i=1,2$.  Then, $\b^{\rho_1} \le \b^{\rho_2}$ on $G$ if $ \rho_1 \, q^{\li} + \alpha \sqrt{1-\rho_1^2} \le \rho_2 \, q^{\li} + \alpha \sqrt{1-\rho_2^2} $ for all $t \in [0, T]$.
\end{cor}

\begin{remark}
A natural question that follows from Corollary \ref{cor_comprho} is when is $f(\rho, t) \triangleq \rho \, q^{\li}(t) + \alpha \sqrt{1-\rho^2}$ decreasing with respect to $\rho$ for $t \in [0, T]$?  Suppose that $\rho > 0$, which is what one expects between the insured and reference populations.  If $f$ is decreasing with respect to $\rho > 0$, then greater positive correlation will lead to a lower per-contract price, an intuitively pleasing result.  It is straightforward to show that $f$ decreases with respect to $\rho$ if and only if
\begin{equation}
\rho > \frac{\qli}{\sqrt{\alpha^2 + \left( \qli \right)^2}}.
\end{equation}
This inequality holds automatically if $\qli < 0$, that is, if the mortality derivative is a so-called {\it natural hedge}, which we discuss more fully in Remark \ref{rem:4_3} below. When $\qli > 0$, it holds for $\rho$ in a neighborhood of 1.
\end{remark}

We have the following special case of Corollary \ref{cor_comprho}.

\begin{cor}\label{cor:Pn}
Suppose $q^{\li}$ is independent of $\li$. If  $\rho \, \qli + \alpha \sqrt{1-\rho^2} < \alpha$, then the limiting price per risk in which hedging is allowed is less than the limiting price with no hedging ($\rho=0$).
\end{cor}

\begin{remark} \label{rem:4_3}
In particular, when $\qli$ is negative (and $\rho$ is positive), the unit price of the contract is reduced by hedging, as demonstrated in Corollary \ref{cor:Pn}.  Since the correlation is usually positive, a mortality derivative with a negative market price of risk $\qli$, that is, a natural hedge, is preferred.  An example of a natural hedge is life insurance, as discussed in \cite{Young2008Pricing}, although strictly speaking this insurance product is not  a mortality derivative traded in the financial market.  Both \cite{BayraktarYoung2007Hedging} and \cite{CoxLin2007Natural} proposed hedging pure endowment or life annuity contracts with life insurance.
\end{remark}

%-----------------------------------------------------------------------------------------------------------------------------------------------------------------------------------------------------------------------------------------------------

\section{\label{sec:numerical}Numerical Example}

In this section, we demonstrate our result with numerical examples. We assume that the risk-free rate of return $r$ is constant and focus on the effect of the correlation $\rho$  and the market price of mortality risk $q^{\li}$. We also assume that the market price of mortality risk is constant, and, thereby, is automatically independent of $\li$.  In this case, $P^{(n)}$ and  $\lim_{n\to \infty}\frac{1}{n}P^{(n)}$ do not depend on $\li$, as we prove in Properties \ref{ppt_psin_simp} and \ref{ppt_beta_simp}. Moreover, we assume that the hazard rate $\lp$ follows the process in \eqref{lp} with  $a^P$ and $b^P$ constant.  We compute the price for a single contract, $P(r, \lp, t) = e^{-r(T-t)} \, \psi(\lp, t)$, and the limiting price per contract for arbitrarily many insureds, $\lim_{n\to \infty} \frac{1}{n}P^{(n)}(r, \lp, t) = e^{-r(T-t)} \, \b(\lp, t)$, and we use the following parameter values:
\begin{itemize}
\item  The pure endowment contract matures in $T = 10$ years. 
\item The constant riskless rate of return is $r = 0.04$.
\item The drift of the hazard rate is $a^P = 0.04$.
\item The volatility of the hazard rate is $b^P = 0.1$.
\item The minimum hazard rate of the insured individuals is $\ulp = 0.02$.
\item The risk parameter is $\alpha = 0.1$.
\end{itemize}
See Section \ref{sec:appendix} for the algorithm that we use to compute $\psi$ and $\beta$.

In Figure \ref{fig1}, for a variety of values of the market price of mortality risk $\qli$, we  present the price of a single-life contract $P$ and the limiting price per contract $\lim_{n\to \infty}\frac{1}{n}P^{(n)}$.  It follows from Theorem \ref{thm:limcon} that, given a positive correlation $\rho > 0$, the limiting unit price of a pure endowment is greater with a greater market price of mortality risk $\qli$, and the second set of graphs in Figure \ref{fig1} demonstrates this result.  Notice that  the price of the unhedged contract is the price with $\rho = 0$.  Since in the pricing mechanism, we hedge the volatility with the mortality derivative as much as possible to reduce the variability of our hedging portfolio, a large value of $\qli$ could lead to a higher contract price than that of an unhedged one.  Observe this numerically in graphs in Figure \ref{fig1}.

In Remark \ref{rem:4_1}, we concluded that if $q^{\li} < \alpha$ and if $\rho = 1$, then the limiting price per contract is less than the limiting price per contract of an unhedged portfolio of pure endowments.  This result is supported by our numerical work; indeed, the curve for $q^{\li} = 0.15$ in the second set of graphs lies above the unhedged price of approximately $0.343$, the price when $\rho = 0$.

In that same remark, we noted that for a single-life contract, we cannot conclude that the price with hedging will be smaller than the price without hedging, even when $\qli < \alpha$.  This conclusion is also supported by our numerical work; indeed, the curve for $\qli = 0.09$ in the first set of graphs lies above the unhedged price of approximately $0.435$, the price when $\rho = 0$.  
%Also, from this set of graphs, we learn that the price $P$ is not necessarily monotone with respect to $\qli$.

Figure \ref{fig1} also demonstrates the relation between the unit price of a contract and the correlation $\rho$.  Take the limiting price per contract $\lim_{n\to \infty} \frac{1}{n}P^{(n)}$ with $\qli=0.05$, for example.  When $\rho=1$, hedging is preferred to not hedging, in terms of reducing the price of the contract.   By contrast, when $\rho<1$, that is, the two mortality rates $\lp$ and $\li$ are not perfectly correlated, hedging may increase the unit price of the contract such as the case when $\rho=0.8$. This observation indicates that the population basis risk, which is the risk due to the mismatch of the insured population and the reference population, diminishes the effectiveness of hedging. This mismatch, or equivalently, a correlation $\rho<1$,  may lead to a higher unit price for the hedged contract.  See \cite{CoughlanEpsteinSinha2007} for discussion of population basis risk. 

%-----------------------------------------------------------------------------------------------------------------------------------------------------------------------------------------------------------------------------------------------------
\section{\label{sec:conclusion}Conclusion}

In this paper, we developed a pricing mechanism for pure endowments, assuming that the issuing company hedges its pure endowment risk with bonds and mortality derivatives, and requires compensation for the unhedged part of mortality risk  in the form of a pre-specified instantaneous Sharpe ratio. In our model, we took the hazard rates of the insured population and reference population, as well as the interest rates, to be stochastic.  We derived the pricing formulae for the hedged contracts on single life and on multiple conditionally independent lives.   Also, we obtained the pricing formula for the limiting price per pure endowment contract as the number of the insureds in the portfolio goes to infinity. In each case, the price solves a PDE, and we analyzed these  PDE and thereby determined  properties of the prices of the hedged pure endowments.  The limiting price per contract  solves a linear PDE and represent this value as an expectation with respect to an equivalent martingale measure.  We noted that, in the limiting case, the mortality risk inherent in the pure endowment is fully hedged by the mortality derivative when the correlation between the two hazard rates $\lp$ and $\li$ is $1$. 

To investigate the factors that affect the effectiveness of hedging, we devoted our attention to the market price of the reference mortality risk $\qli$ and the correlation $\rho$ between $\lp$ and $\li$.  Since the correlation $\rho$ is more likely to be positive in reality, we focused on the case for which $\rho\ge0$ during our discussion (and especially in our numerical work) and assumed that the market price of the mortality risk $\qli$ is  independent of $\li$.  We found that hedging with a mortality derivative requiring a negative market price of mortality risk always reduces the price of the contract. This result is consistent with the conclusions in  \cite{BayraktarYoung2007Hedging} and \cite{CoxLin2007Natural} that hedging pure endowments (or life annuities) with life insurance reduces the price of the former.

%In spite of the complicated relation between the effectiveness of hedging and the two factors $\qli$ and $\rho$, we obtained the necessary condition that $\qli<\alpha$  for hedging to be effective in terms of reducing the price of the pure endowment.  In other words, hedging with the mortality derivatives benefits the insured only if the market requires less compensation to take the mortality risk than the insurance company does.  {\bf Wang Ting, are we saying that this is a necessary condition because it was necessary and sufficient in the limiting case?}

For the limiting case, we reached a more straightforward conclusion, as we discuss in Remark \ref{rem:4_1}.  Specifically, if $\rho=1$, the condition that $\qli<\alpha$ guarantees a reduction in the per-contract price through hedging.  However, if $\rho<1$, it is possible that hedging with the mortality derivatives increases the price of the contract even if this condition is satisfied.  This result reflects the significance of $\rho$ on the effectiveness of hedging.  We also found that, in our numerical work, hedging with the mortality derivatives is less effective in reducing the variance of the hedging portfolio for pure endowments of a finite number of individuals. 

Our results suggest that, to make it efficient for underwriters to hedge mortality risk and thereby benefit the insured, transparent design of mortality indices and mortality derivatives is essential.   Reducing the market price of the mortality risk $\qli$ is also critical. Therefore, it is important to build up a liquid mortality market and provide more flexible mortality-linked securities in order to reduce $\qli$. 

In our paper, we only investigated the prices of pure endowments and assumed that the mortality derivative is a $q$-forward.  However, we believe that the main qualitative insights will hold in general.

%-----------------------------------------------------------------------------------------------------------------------------------------------------------------------------------------------------------------------------------------------------
\section{Appendix \label{sec:appendix}}

In this section, we present an algorithm for numerically computing $\psi$.  Recall that in our numerical example, we assume that $q^{\li}$, the market price of mortality risk $\li$, is a constant, as well as $a^P$ and $b^P$.  Then, equation \eqref{eqn_psi} becomes 
\begin{equation}\label{eqn_psi_const}
\begin{cases}
\psi_t  + \left[ a^P - \rho \, \qli b^P \right] \lpa  \psi_{\lp}  + \dfrac{1}{2} \left( b^P \right)^2 \lpa^2 \psi_{\lp \lp}  - \lp \psi \\
\quad = -\alpha \sqrt{\left(1-\rho^2 \right) \left(b^P\right)^2 \lpa^2 \psi_{\lp}^2+\lp \psi^2},\\
\psi\left(\lp, T \right) = 1.
\end{cases}
\end{equation} 

Next, we describe our numerical scheme to compute ${\psi}$.
\begin{description} 
\item[Transformation] Define $\tau=T-t$,  $y=\ln\left( \lp-\ulp \right)$, and $\hat{\psi}\left(y, \tau \right)=\psi\left(\lp, t\right)$.  By \eqref{eqn_psi_const},  $\hat{\psi}$ solves 
\begin{equation} \label{eqn_psi_num}
\begin{cases}
\hat{\psi}_\tau = \hat{a} \, \hat{\psi}_{y} + \dfrac{1}{2} \left( b^P \right)^2 \hat{\psi}_{yy} - \left(e^y+\ulp\right) \hat{\psi} + \alpha \sqrt{\left(1-\rho^2 \right) \left(b^P\right)^2 \hat{\psi}_{y}^2+\left(e^y+\ulp\right) \hat{\psi}^2},\\
\hat{\psi}\left(y, 0\right)=1,
\end{cases}
\end{equation}
in which $\hat{a} = a^{P}-\rho \,  \qli b^P - \frac{1}{2} \left( b^P \right)^2$. 
\item[Boundary Condtions] While equation \eqref{eqn_psi_num} for $\hat{\psi}$ is defined in the domain $\R \times [0,T]$,  we solve it numerically in the domain $[-M, M] \times [0, T]$  such that $e^{-M}$ is approximately zero.  Therefore, we require boundary conditions at $y=\pm M$. 
\begin{enumerate}
\item If $\lp_t = \ulp$, then $\lp_s = \ulp$ for all $s \in [t, T]$. From equation \eqref{eqn_psi_const}, we have that $\psi\left(\ulp,t \right)=\exp\left\{-\left( \ulp-\alpha\ \sqrt{\ulp}\right)(T-t)\right\}$. Thus, it is reasonable to set the boundary condition at $y=-M$ to be $\hat{\psi}\left(-M,\tau \right)=\exp\left\{-\left( \ulp-\alpha\ \sqrt{\ulp}\right)\tau\right\}$.
\item If $\lp_t$ is very large, we expect the individual to die immediately, so the value of the pure endowment is approximately $0$. Thus, we set the boundary condition at $y=M$ to be $\hat{\psi}\left(M,\tau \right)= 0$.
\end{enumerate}
\item [Finite Difference Scheme] We discretize the differential equation \eqref{eqn_psi_num} and get a corresponding difference equation as follows:
\begin{enumerate}
\item Choose the step sizes of $y$ and $\tau$ as $h$ and $k$, respectively, so that $I=2M/h$ and $J=T/k$ are integers.
\item Define $y_i=-M+ih$, $\tau_j=jk$, and $\hat{\psi}_{i,j}=\hat{\psi}\left(y_i,\tau_j\right)$, for $i=0,1,\dots, I$ and $j=0, 1, \dots , J$.  
\item  We use a backward difference in time, central differences in space, and a forward difference for the square-root term. Therefore, we have the following expressions:
\begin{equation}
\begin{cases}
\hat{\psi}_{\tau}\left(y_i,\tau_j\right)=\dfrac{\hat{\psi}_{i,j+1}-\hat{\psi}_{i,j}}{k} + \mathcal{O}(k),\\
\hat{\psi}_{y}\left(y_i,\tau_j\right)=\dfrac{\hat{\psi}_{i+1,j+1}-\hat{\psi}_{i-1,j+1}}{2h} + \mathcal{O}(h^2),\\
\hat{\psi}_{yy}\left(y_i,\tau_j\right)=\dfrac{\hat{\psi}_{i+1,j+1}-2\hat{\psi}_{i,j+1}+\hat{\psi}_{i-1,j+1}}{h^2} + \mathcal{O}(h^2).\\
\end{cases}
\end{equation}
Also, for the non-linear term in \eqref{eqn_psi_num}, we have
\begin{equation}
\begin{split}
&\quad \sqrt{\left(1-\rho^2 \right) \left(b^P\right)^2 \hat{\psi}_{y}^2+\left(e^y+\ulp\right) \hat{\psi}^2}\\
&=  \sqrt{\left(1-\rho^2 \right) \left(b^P\right)^2 \left(\dfrac{\hat{\psi}_{i+1,j}-\hat{\psi}_{i-1,j}}{2h}\right)^2+\left(e^{y_i}+\ulp\right) \hat{\psi}_{i,j}^2}+\mathcal{O}(h).
\end{split}
\end{equation}
Therefore, we approximate \eqref{eqn_psi_num} to order $\mathcal{O}(k+h)$ with the following difference equation:
\begin{equation}\label{eqn_psihat}
\begin{split}
\dfrac{\hat{\psi}_{i,j+1} - \hat{\psi}_{i,j}}{k} &= \hat{a} \, \dfrac{\hat{\psi}_{i+1,j+1}-\hat{\psi}_{i-1,j+1}}{2h} +\dfrac{1}{2}\left(b^P\right)^2\dfrac{\hat{\psi}_{i+1,j+1} -2\hat{\psi}_{i,j+1}+\hat{\psi}_{i-1,j+1}}{h^2} \\
& \quad -\left(e^{y_i}+\ulp\right)\hat{\psi}_{i,j+1} + \alpha A_{i,j},\\
\end{split}
\end{equation}
in which
\begin{equation}\label{Aij}
A_{i,j}=\sqrt{\left(1-\rho^2 \right) \left(b^P\right)^2 \left(\dfrac{\hat{\psi}_{i+1,j}-\hat{\psi}_{i-1,j}}{2h}\right)^2+\left(e^{y_i}+\ulp\right) \hat{\psi}_{i,j}^2} \; .
\end{equation}
If we define $a = \hat{a} \, \frac{k}{2h} - \left(b^P\right)^2 \frac{k}{2h^2}$, $b=1+\left(b^P\right)^2 \frac{k}{h^2} + k \, \ulp$, and $c=-\hat{a} \, \frac{k}{2h} - \left(b^P\right)^2 \frac{k}{2h^2}$,  then \eqref{eqn_psihat} becomes 
\begin{equation}\label{num_temp1}
a \hat{\psi}_{i-1,j+1}+\left(b+k e^{y_i}\right)\hat{\psi}_{i,j+1}+c\hat{\psi}_{i+1,j+1}=\hat{\psi}_{i,j}+\alpha k A_{i,j},
\end{equation}
for $i=1, 2, \dots, I-1$ and $j=0, 1, \dots, J-1$, with the following boundary conditions:
\begin{equation}\label{num_temp2}
\begin{cases}
\left(b+k e^{y_1}\right)\hat{\psi}_{1,j+1}+c\hat{\psi}_{2,j+1}=\hat{\psi}_{1,j}+\alpha k A_{1,j}-a e^{-\left( \ulp-\alpha\ \sqrt{\ulp}\right)(j+1)k}\\
a\hat{\psi}_{I-2,j+1} + \left(b+k e^{y_{I-1}}\right)\hat{\psi}_{I-1,j+1}=\hat{\psi}_{I-1,j}+\alpha k A_{I-1,j},\\
\end{cases}
\end{equation}
for $j=0, 1, \dots, J-1$.  It is convenient to write equations  \eqref{num_temp1}-\eqref{num_temp2}  in matrix form as 
\begin{equation}\label{num_matrix}
{\bf M \hat{\Psi}}_{j+1}={\bf\hat{\Psi}}_{j}+\alpha k {\bf A}_{j}-\left[a e^{-\left( \ulp-\alpha\ \sqrt{\ulp}\right)(j+1)k},0,\dots,0\right]^t. 
\end{equation}
for $j=0, 1, \dots, J-1$, in which the superscript $t$ represents matrix transpose.  In the equation above, ${\bf \hat{\Psi}}_j=\left[\hat{\psi}_{1,j}, \hat{\psi}_{2,j}, \dots, \hat{\psi}_{I-1,j}\right]^t$ and ${\bf A}_j=\left[A_{1,j}, A_{2,j}, \dots, A_{I-1,j}\right]^t$ with $A_{i,j}$ defined in \eqref{Aij}. The matrix ${\bf M}$ is a tri-diagonal matrix with the sub-diagonal identically $a$, with the main diagonal $b+ke^{y_1}, b+ke^{y_2}, \dots, b+ke^{y_{I-1}}$, and with the super-diagonal identically $c$. 
\item Begin with the initial condition $\hat{\psi}_{i,0}=1$, for $i=1,2, \dots, I-1$, or equivalently, ${\bf \hat{\Psi}}_0={\bf 1}$, in which ${\bf 1}$ is an $(I-1)\times1$ column vector of $1$s. Then, solve \eqref{num_matrix} repeatedly for $j = 0, 1, \dots, J-1$ until we reach ${\bf \hat{\Psi}}_J$. 
\end{enumerate}
\end{description}
One can modify this algorithm to compute $\psi^{(n)}$ for any $n>1$ and the limiting result $\beta = \lim_{n\to \infty}\psi^{(n)}$. Computing the latter is particularly straightforward because $\b$ solves a linear PDE.

\newpage
%-----------------------------------------------------------------------------------------------------------------------------------------------------------------------------------------------------------------------------------------------------

\bibliography{MLS}

\begin{thebibliography}{27}
\providecommand{\natexlab}[1]{#1}
\providecommand{\url}[1]{\texttt{#1}}
\expandafter\ifx\csname urlstyle\endcsname\relax
  \providecommand{\doi}[1]{doi: #1}\else
  \providecommand{\doi}{doi: \begingroup \urlstyle{rm}\Url}\fi

\bibitem[Bauer et~al.(2010)Bauer, B{\"o}rger, and Ru{\ss}]{bauer2009pricing}
D.~Bauer, M.~B{\"o}rger, and J.~Ru{\ss}.
\newblock On the pricing of longevity-linked securities.
\newblock \emph{Insurance: Mathematics and Economics}, 46:\penalty0 139--149,
  2010.

\bibitem[Bayraktar and Ludkovski(2009)]{BayraktarLudkovski2009Relative}
E.~Bayraktar and M.~Ludkovski.
\newblock Relative hedging of systematic mortality risk.
\newblock \emph{North American Actuarial Journal}, 13\penalty0 (1):\penalty0
  106--140, 2009.

\bibitem[Bayraktar and Young(2007{\natexlab{a}})]{BayraktarYoung2007Hedging}
E.~Bayraktar and V.~R. Young.
\newblock Hedging life insurance with pure endowments.
\newblock \emph{Insurance: Mathematics and Economics}, 40\penalty0
  (3):\penalty0 435--444, 2007{\natexlab{a}}.

\bibitem[Bayraktar and Young(2007{\natexlab{b}})]{BayraktarYoung2007Pricing}
E.~Bayraktar and V.~R. Young.
\newblock Pricing options in incomplete equity markets via the instantaneous
  sharpe ratio.
\newblock \emph{Annals of Finance}, 4\penalty0 (4):\penalty0 399--429,
  September 2007{\natexlab{b}}.

\bibitem[Bayraktar et~al.(2009)Bayraktar, Milevsky, Promislow, and
  Young]{BayraktarMilevskyPromislow2008Valuation}
E.~Bayraktar, M.~A. Milevsky, S.~D. Promislow, and V.~R. Young.
\newblock Valuation of mortality risk via the instantaneous sharpe ratio:
  Applications to life annuities.
\newblock \emph{Journal of Economic Dynamics and Control}, 33\penalty0
  (3):\penalty0 676--691, 2009.

\bibitem[Biffis(2005)]{Biffis2005Affine}
E.~Biffis.
\newblock Affine processes for dynamic mortality and actuarial valuatons.
\newblock \emph{Insurance: Mathematics and Economics}, 2005.

\bibitem[Bj\"ork(2004)]{Bjork2004Arbitrage}
T.~Bj\"ork.
\newblock \emph{Arbitrage Theory in Continuous Time}.
\newblock Oxford University Press, 2nd edition, 2004.

\bibitem[Blake and Burrows(2001)]{BlakeBurrows2001}
D.~Blake and W.~Burrows.
\newblock Survivor bonds: Helping to hedge mortality risk.
\newblock \emph{The Journal of Risk and Insurance}, 68:\penalty0 339--348,
  2001.

\bibitem[Blake et~al.(2006)Blake, Cairns, and Dowd]{blake2006living}
D.~Blake, A.~J.~G. Cairns, and K.~Dowd.
\newblock Living with mortality: longevity bonds and other mortality-linked
  securities.
\newblock \emph{British Actuarial Journal}, 12\penalty0 (1):\penalty0 153--197,
  2006.

\bibitem[Cairns et~al.(2006{\natexlab{a}})Cairns, Blake, and
  Dowd]{CairnsBlakeDowd2006}
A.~J.~G. Cairns, D.~Blake, and K.~Dowd.
\newblock A two-factor model for stochastic mortality with parameter
  uncertainty: Theory and calibration.
\newblock \emph{The Journal of Risk and Insurance}, 73:\penalty0 687--718,
  2006{\natexlab{a}}.

\bibitem[Cairns et~al.(2006{\natexlab{b}})Cairns, Blake, and
  Dowd]{CairnsBlakeDowd2006Pricing}
A.~J.~G. Cairns, D.~Blake, and K.~Dowd.
\newblock Pricing death: Frameworks for the valuation and securitization of
  mortality risk.
\newblock \emph{ASTIN Bulletin}, 36:\penalty0 79--120, 2006{\natexlab{b}}.

\bibitem[Coughlan et~al.(2007)Coughlan, Epstein, and
  Sinha]{CoughlanEpsteinSinha2007}
G.~Coughlan, D.~Epstein, and A.~Sinha.
\newblock q-forwards: Derivatives for transferring longevity and mortality
  risk.
\newblock \emph{Pension Advisory Group Report}, 2007.

\bibitem[Cox and Lin(2007)]{CoxLin2007Natural}
S.~H. Cox and Y.~Lin.
\newblock Natural hedging of life and annuity mortality risks.
\newblock \emph{North American Actuarial Journal}, 11\penalty0 (3):\penalty0
  1--15, 2007.

\bibitem[Dahl(2004)]{Dahl2004}
M.~Dahl.
\newblock Stochastic mortality in life insurance: Market reserves and
  mortality-linked insurance contracts.
\newblock \emph{Insurance: Mathematics and Economics}, 35\penalty0
  (1):\penalty0 113--136, 2004.

\bibitem[Dahl and Moller(2006)]{DahlMoller2006Valuation}
M.~Dahl and T.~Moller.
\newblock Valuation and hedging of life insurance liablities with systematic
  mortality risk.
\newblock \emph{Insurance: Mathematics and Economics}, 39\penalty0
  (2):\penalty0 193--217, 2006.

\bibitem[Dowd et~al.(2006)Dowd, Blake, Cairns, and
  Dawson]{DowdBlakeCairnsDawson2006}
K.~Dowd, D.~Blake, A.~J.~G. Cairns, and P.~Dawson.
\newblock Survivor swaps.
\newblock \emph{The Journal of Risk and Insurance}, 73:\penalty0 1--17, 2006.

\bibitem[Lin and Cox(2005)]{LinCox2005Securitization}
Y.~Lin and S.~H. Cox.
\newblock Securitization of mortality risks in life annuities.
\newblock \emph{The Journal of Risk and Insurance}, 72\penalty0 (2):\penalty0
  227--252, Jun 2005.

\bibitem[Milevsky and Promislow(2001)]{MilevskyPromislow2001}
M.~A. Milevsky and S.~D. Promislow.
\newblock Mortality derivatives and the option to annuitize.
\newblock \emph{Insurance: Mathematics and Economics}, 29\penalty0
  (3):\penalty0 299--318, 2001.

\bibitem[Milevsky et~al.(2005)Milevsky, Promislow, and
  Young]{MilevskyPromislowYoung2005Financial}
M.~A. Milevsky, S.~D. Promislow, and V.~R. Young.
\newblock Financial valuation of mortality risk via the instantaneous sharpe
  ratio: Applications to pricing pure endowments.
\newblock \emph{arXiv: 0705.1302v1}, 2005.

\bibitem[Milevsky et~al.(2006)Milevsky, Promislow, and
  Young]{MilevskyPromislowYoung2006Killing}
M.~A. Milevsky, S.~D. Promislow, and V.~R. Young.
\newblock Killing the law of large numbers: Mortality risk premiums and the
  sharpe ratio.
\newblock \emph{Journal of Risk and Insurance}, 73\penalty0 (4):\penalty0
  673--686, 2006.

\bibitem[Milidonis et~al.(2010)Milidonis, Lin, and Cox]{MilidonisLinCox2010}
A.~Milidonis, Y.~Lin, and S.~H. Cox.
\newblock Mortality regimes and pricing.
\newblock \emph{Working Paper}, 2010.

\bibitem[Miltersen and Persson(2005)]{MiltersenPersson2005}
K.~R. Miltersen and S.~A. Persson.
\newblock Is mortality dead? stochastic forward force of mortality rate
  determined by no arbitrage.
\newblock \emph{Working Paper}, 2005.

\bibitem[Schrager(2006)]{Schrager2006Affine}
D.~F. Schrager.
\newblock Affine stochastic mortality.
\newblock \emph{Insurance: Mathematics and Economics}, 38:\penalty0 81--97,
  2006.

\bibitem[Schweizer(2001)]{Schweizer2001A-Guided}
M.~Schweizer.
\newblock A guided tour through quadratic hedging approaches.
\newblock In \emph{Handbooks in Mathematical Finance: Option Pricing, Interest
  Rates and Risk Management}. Cambridge University Press, New York, 2001.

\bibitem[Walter(1970)]{Walter}
W.~Walter.
\newblock \emph{Differential and Integral Inequalities}.
\newblock Springer-Verlag, New York, 1970.

\bibitem[Young(2004)]{Young2004}
V.~R. Young.
\newblock Premium principles.
\newblock In \emph{Encyclopedia of Actuarial Science}. Wiley, New York, 2004.

\bibitem[Young(2008)]{Young2008Pricing}
V.~R. Young.
\newblock Pricing life insurance under stochastic mortality via the
  instantaneous sharpe ratio.
\newblock \emph{Insurance: Mathematics and Economics}, 42:\penalty0 691--703,
  2008.

\end{thebibliography}
\bibliographystyle{abbrvnat}

%-----------------------------------------------------------------------------------------------------------------------------------------------------------------------------------------------------------------------------------------------------

\begin{figure}
	\centering
		\includegraphics[angle=90, width=0.8\textwidth]{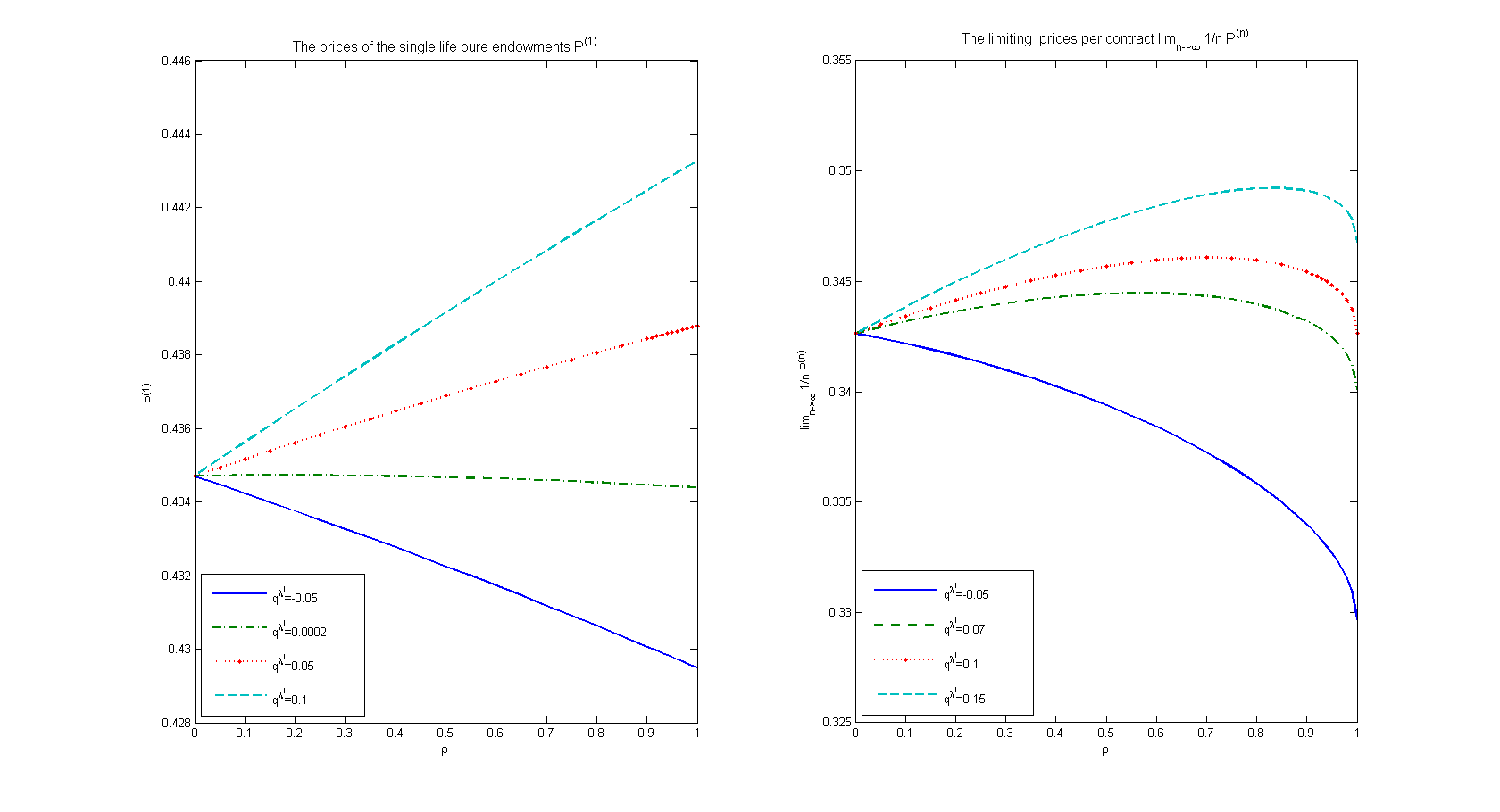}
	\label{fig1}
	\caption{}
\end{figure}

\end{document}